\documentclass[11pt, letterpaper]{amsart}
\usepackage{graphicx, amssymb}
\usepackage{amsmath}
\usepackage[round]{natbib}
\usepackage{paralist,color}
\usepackage{verbatim}

\newtheorem{thm}{Theorem}[section]
\newtheorem{cor}[thm]{Corollary}
\newtheorem{lem}[thm]{Lemma}
\newtheorem{prop}[thm]{Proposition}
\theoremstyle{definition}
\newtheorem{defn}[thm]{Definition}

\newtheorem{ass}[thm]{Assumption}

\theoremstyle{remark}
\newtheorem{rem}[thm]{Remark}
\newtheorem{exa}[thm]{Example}

\numberwithin{equation}{section}

\newcommand{\set}[1]{\left\{#1\right\}}
\newcommand{\Real}{\mathbb R}
\newcommand{\Natural}{\mathbb N}
\newcommand{\eps}{\varepsilon}

\newcommand{\such}{\, | \, }
\newcommand{\prob}{\mathbb{P}}
\newcommand{\Exp}{\mathcal E}

\newcommand{\qprob}{\mathbb{Q}}
\newcommand{\expec}{\mathbb{E}}

\newcommand{\F}{\mathcal{F}}

\newcommand{\B}{\mathcal{B}}
\newcommand{\cadlag}{c\`adl\`ag}

\newcommand{\pare}[1]{\left(#1\right)}
\newcommand{\bra}[1]{\left[#1\right]}
\newcommand{\dbra}[1]{[\kern-0.15em[ #1 ]\kern-0.15em]}
\newcommand{\dbraco}[1]{[\kern-0.15em[ #1 [\kern-0.15em[}
\newcommand{\dbraoc}[1]{]\kern-0.15em] #1 ]\kern-0.15em]}

\newcommand{\indic}{1}

\newcommand{\X}{\mathcal{X}}

\newcommand{\cL}{\mathcal{L}}
\newcommand{\fR}{\mathfrak{R}}

\newcommand{\nada}[1]{}
\newcommand{\tbF}{\widetilde{\mathbf{F}}}
\newcommand{\tF}{\widetilde{\F}}
\newcommand{\tS}{\widetilde{S}}
\newcommand{\tX}{\widetilde{X}}
\newcommand{\tY}{\widetilde{Y}}
\newcommand{\tsmY}{\widetilde{y}}

\title{Abstract, Classic, and Explicit Turnpikes}

\author[]{Paolo Guasoni}
\thanks{Paolo Guasoni is partially supported by the ERC (278295), NSF (DMS-0807994, DMS-1109047), SFI (07/MI/008, 07/SK/M1189, 08/SRC/FMC1389), and FP7 (RG-248896).}
\address[Paolo Guasoni]{School of Mathematical Sciences,
Dublin City University,
Glasnevin, Dublin 9,
Ireland. Department of Mathematics and Statistics,
Boston University,
111 Cummington st,
Boston, MA 02215,
USA}
\email{guasoni@bu.edu}

\author[]{Constantinos Kardaras}
\thanks{Constantinos Kardaras is partially supported by the NSF DMS-0908461.}
\address[Constantinos Kardaras]{Department of Mathematics and Statistics,
Boston University,
111 Cummington st,
Boston, MA 02215,
USA}
\email{kardaras@bu.edu}

\author[]{Scott Robertson}
\address[Scott Robertson]{Department of Mathematical Sciences,
Wean Hall 6113,
Carnegie Mellon University,
Pittsburgh, PA 15213,
USA}
\email{scottrob@andrew.cmu.edu}

\author[]{Hao Xing}
\address[Hao Xing]{Department of Statistics,
London School of Economics and Political Science,
10 Houghton st,
London, WC2A 2AE,
UK}
\email{h.xing@lse.ac.uk}

\begin{document}

\begin{abstract}
Portfolio turnpikes state that, as the investment horizon increases, optimal portfolios for generic utilities converge to those of isoelastic utilities.
This paper proves three kinds of turnpikes.  In a general semimartingale setting, the \emph{abstract} turnpike states that optimal final payoffs and portfolios converge under their myopic probabilities.
In diffusion models with several assets and a single state variable,  the \emph{classic} turnpike demonstrates that optimal portfolios converge under the physical probability; meanwhile the \emph{explicit turnpike} identifies the limit of finite-horizon optimal portfolios as a long-run myopic portfolio defined in terms of the solution of an ergodic HJB equation.
\end{abstract}


\keywords{Portfolio Choice, Incomplete Markets, Long-Run, Utility Functions, Turnpikes}

\maketitle

\section{Introduction}
In the theory of portfolio choice, ruled by particular and complicated results, turnpike theorems are happy exceptions -- general and simple.
Loosely defined, these theorems state that, when the investment horizon is distant, the optimal portfolio of \emph{any} investor approaches that of an investor with isoelastic utility, suggesting that \emph{for long-term investments, only isoelastic utilities matter}.

This paper proves turnpike theorems in a general framework, which include discrete and continuous time, and nest diffusion models with several assets,  stochastic drifts, volatilities, and interest rates. The paper departs from the existing literature, in which either asset returns are independent over time, or markets are complete. It is precisely when both these assumptions fail that portfolio choice becomes most challenging -- and turnpike theorems are most useful.

Our results have three broad implications. First, turnpike theorems are a powerful tool in portfolio choice, because they apply not only when optimal portfolios are myopic, but also when the \emph{intertemporal hedging} component is present. Finding this component is the central problem of portfolio choice, and the only tractable but non trivial analysis is based on isoelastic utilities, combined with long horizon asymptotics. Turnpike theorems make this analysis relevant for a large class of utility functions, and for large but finite horizons.

Second, we clarify the roles of preferences and market structure for turnpike results. Under regularity conditions on utility functions, we show that an \emph{abstract turnpike} theorem holds regardless of market structure, as long as utility maximization is well posed, and longer horizons lead to higher payoffs.
This abstract turnpike yields the convergence of optimal portfolios to their isoelastic limit under the \emph{myopic} probability $\prob^T$, which changes with the horizon. Market structure becomes crucial to pass from from the abstract to the \emph{classic turnpike} theorem, in which convergence holds under the physical probability $\prob$.

Third, in addition to the classical version, we prove a new kind of result, the \emph{explicit turnpike}, in which the limit portfolio is identified as the \emph{long-run optimal} portfolio, that is the solution to an ergodic Hamilton Jacobi Bellman equation. This result offers the first theoretical basis to the long-standing practice of interpreting solutions of ergodic HJB equations as long-run limits of utility maximization problems\footnote{This interpretation underpins the literature on \emph{risk-sensitive control}, introduced by \citet*{MR1358100}, and applied to optimal portfolio choice by \citeauthor{MR1802598} \citeyearpar{MR1802598,MR1910647}, \citeauthor{MR1882297} \citeyearpar{MR1790132,MR1882297}, \citeauthor{MR1890061} \citeyearpar{MR1890061,MR1882294} among others.}.
 We show that this intuition is indeed correct for a large class of diffusion models, and that its scope includes a broader class of utility functions.

Portfolio turnpikes start with the work of \citet*{mossin1968optimal} on affine risk tolerance ($-U'(x)/U''(x)= a x+b$), which envisions many of the later developments. In his concluding remarks, he writes: \emph{``Do any of these results carry over to arbitrary utility functions? They seem reasonable enough, but the generalization does not appear easy to make. As one usually characterizes those problems one hasn't been able to solve oneself: this is a promising area for future research''.}

\citet*{leland1972turnpike} coins the expression \emph{portfolio turnpike}, extending Mossin's result to larger classes of utilities, followed by  \citet*{ross1974portfolio} and \citet*{hakansson1974convergence}. \citet*{MR736053} prove a necessary and sufficient condition for the turnpike property. As in the previous literature, they consider discrete time models with independent returns.
\citet*{cox1992continuous} prove the first turnpike theorem in continuous time, using contingent claim methods. \citet*{MR1629559} extends their results to include consumption, and \citet*{MR1805320} obtain similar results using viscosity solutions.
\citet*{dybvig1999portfolio} dispose of the assumption of independent returns, proving a turnpike theorem for complete markets in the Brownian filtration, while \citet*{detemple2010dynamic} obtain a portfolio decomposition formula for complete markets, which allows to compute turnpike portfolios in certain models.

In summary, the literature either exploits independent returns, which make dynamic programming attractive, or complete markets, which make martingale methods convenient. Since market completeness and independence of returns have a tenuous relation, neither of these concepts appears to be central to turnpike theorems. We confirm this intuition, by relaxing both assumptions in this paper.

The main results are in section \ref{sec: myopic}, which is divided into two parts. The first part shows the conditions leading to the abstract turnpike, whereby optimal final payoffs and portfolios converge under the myopic probabilities $\prob^T$. Regarding preferences (Assumption \ref{ass: utility}), we require a marginal utility that is asymptotically isoelastic as wealth increases \eqref{ass: conv}, and a well-posed utility maximization problem.

The second part of section \ref{sec: myopic} states the classic and explicit turnpike theorems for a class of diffusion models with several assets, but with a single state variable driving expected returns, volatilities and interest rates. Under further well-posedness assumptions, we show a classic turnpike theorem, in which optimal portfolios of generic utility functions converge to their isoelastic counterparts. The same machinery leads to the explicit turnpike, in which optimal portfolios for a generic utility and a finite horizon converge to the long-run optimal portfolio, that is the solution of an ergodic HJB equation.
We conclude section 2 with an application to target-date retirement funds, which shows that a fund manager who tries to maximize the weighted welfare of participants -- like a social planner -- tends to act on behalf of the least risk averse investors.

Section 3 contains the proofs of the abstract turnpike, while the classic and the explicit turnpike for diffusions are proved in section 4. The first part of section 3 proves the convergence of the ratio of final payoffs, while the second part derives the convergence of wealth processes.
Section 4 studies the properties of the long-run measure and the value function, and continues with the convergence of densities and wealth processes, from which the classic and explicit turnpikes follow.

In conclusion, this paper shows that turnpike theorems are an useful tool to make portfolio choice tractable, even in the most intractable setting of incomplete markets combined with stochastic investment opportunities. Still, these results are likely to admit extensions to more general settings, like diffusions with multiple state variables, and processes with jumps. As gracefully put by Mossin, this is a promising area for future research.

\section{Main Results}\label{sec: myopic}

This section contains the statements of the main results and their implications. The first subsection states an abstract version of the turnpike theorem, which focuses on payoff spaces and wealth processes, without explicit reference to the structure of the underlying market. In this setting, asymptotic conditions on the utility functions and on wealth growth imply that, as the horizon increases, optimal wealths and optimal portfolios converge to their isoelastic counterparts.

The defining feature of the abstract turnpike is that convergence takes place under a family of \emph{myopic} probability measures that change with the horizon. By contrast, in the \emph{classic} turnpike the convergence holds under the physical probability measure.
Thus, passing from the abstract to the classic turnpike theorem requires the convergence of the myopic probabilities, which in turn commands additional assumptions.
The second subsection achieves this task for a class of diffusion models with several risky assets, and with a single state variable driving investment opportunities. This class nests several models in the literature, and allows for return predictability, stochastic volatility, and stochastic interest rates.

The \emph{explicit} turnpike -- stated at the end of the second subsection -- holds for the same class of diffusion models. While in the abstract and classical turnpikes the benchmark is the optimal portfolio for isoelastic utility, but with the same finite horizon, in the explicit turnpike the benchmark is the long-run optimal portfolio, that is the optimal portfolio for asymptotic expected utility.

\subsection{The Abstract Turnpike}

Consider two investors, one with Constant Relative Risk Aversion (henceforth CRRA) equal to $1-p$ (i.e. power utility $x^p/p$ for $0\ne p<1$ or logarithmic utility $\log x$ for $p=0$), the other with a generic utility function $U: \Real_+ \rightarrow \Real$. The marginal utility ratio $\fR(x)$ measures how close $U$ is to the reference utility:
\begin{equation}\label{eq: marginal ratio function}
 \fR(x):= \frac{U'(x)}{x^{p-1}}, \quad x>0.
\end{equation}

\begin{ass}\label{ass: utility}
The utility function $U: \Real_+ \rightarrow \Real$ is continuously differentiable, strictly increasing, strictly concave, and satisfies the Inada conditions $U'(0)=\infty$ and $U'(\infty)=0$. The marginal utility ratio satisfies:
\begin{equation}\label{ass: conv}
\lim_{x \uparrow \infty} \fR(x) =1. \tag{CONV}
\end{equation}
\end{ass}
Condition \eqref{ass: conv} means that investors have similar marginal utilities when wealth is high, and is the basic assumption on \emph{preferences} for turnpike theorems \citep{dybvig1999portfolio,MR1805320}.

Both investors trade in a frictionless market with one safe and $d$ risky assets. Consider a filtered probability space $(\Omega, (\F_t)_{t\in[0,T]}, \F, \prob)$, where $(\F_t)_{t\ge 0}$ is a right-continuous filtration. The safe asset, denoted by $(S^0_t)_{t\ge 0}$ and the risky assets $(S^i_t)_{t\ge 0}^{1\le i\le d}$ satisfy the following:
\begin{ass}\label{ass: safe asset}
$S^0$ has RCLL (right-continuous, left-limited) paths, and there exist two deterministic functions $\underline S^0, \overline S^0 : \Real_+ \mapsto \Real_+$, such that
$0<\underline S^0_t\le S^0_t \le \overline S^0_t$ for all $t>0$ and
\begin{equation}\label{ass: growth}
  \lim_{T\rightarrow \infty} \underline{S}^0_T = \infty. \tag{GROWTH}
\end{equation}
\end{ass}
This condition means that growth continues over time, and is the main \emph{market} assumption in the turnpike literature. It implies, that the riskless discount factor declines to zero in the long run.
Now, denote the discounted prices of risky assets by $\widetilde{S}^i = S^i / S^0$ for $i=1, \ldots, d$, and set $\widetilde{S} = (\widetilde{S}^1_t)^{1\le i\le d}_{t\ge 0}$. The following assumption is equivalent to the absence of arbitrage, in the sense of No Free Lunch with Vanishing Risk \citep{DS93,DS98}. In particular, up to a null set, $S = (S^i_t)_{t\ge 0}^{1\le i\le d}$ is a $\Real^d$-valued semimartingale with RCLL paths.
\begin{ass}\label{ass: NFLVR}
For all $T \in \Real_+$, there exists a probability $\qprob^T$ that is equivalent to $\prob$ on $\F_T$ and such that $\widetilde{S}$ is a (vector) sigma-martingale.
\end{ass}

Starting from unit capital, each investor trades with some admissible strategy $H$, that is a $S$-integrable and $\F$-predictable $\Real^d$-valued process, such that $\widetilde{X}_t^H:= 1 + \int_0^t H_u \,d \widetilde{S}_u \ge 0$ $\prob$-a.s. for all $t\ge 0$.
Denote a wealth process by $X^H = S^0 \widetilde{X}^H$, and their class by $\mathcal{X}^T:= \{X^H: H \text{ is } T\text{-admissible}\}$.

Both investors seek to maximize the expected utility of their terminal wealth at some time horizon $T$. Using the index $0$ for the CRRA investor, and $1$ for the generic investor, their optimization problems are:
\begin{equation}\label{def: value u}
u^{0,T} = \sup_{X \in \X^T} \expec^{\prob}\bra{X_T^p/p},
\qquad
u^{1, T}= \sup_{X \in \X^T} \expec^{\prob}\bra{U\pare{X_T}}.
\end{equation}
The next assumption requires that these problems are well-posed. It holds under the simple criteria in \citet*[Remark 8]{karatzas.zitkovic.03}.
\begin{ass}\label{ass: wellpose}
For all $T>0$ and $i=0,1$, $u^{i,T}<\infty$.
\end{ass}
\cite{karatzas.zitkovic.03} show that, under Assumptions \ref{ass: utility}-\ref{ass: wellpose}, the optimal wealth processes $X^{i,T}$ exist for $i=0, 1$ and any $T\geq 0$. In addition, $u^{i,T}>-\infty$, because both investors can invest all their wealth in $S^0$ alone, and $S^0_T$ is bounded away from zero by a constant.

The central objects in the abstract turnpike theorem are the ratio of optimal wealth processes and their stochastic logarithms
\begin{equation}\label{def: ratio log}
r^T_{u} := \frac{X^{1,T}_u}{X^{0,T}_u}, \qquad
\Pi^T _u:= \int_0^{u} \frac{dr^T_v}{r^T_{v-}}, \quad \text{ for } u \in [0,T],
\end{equation}
and are well-defined by Remark~\ref{rem: long run measure} below. Moreover, $r_0^T=1$ since both investors have the same initial capital. Define also the \emph{myopic probabilities} $(\prob^T)_{T\geq 0}$ by:
\begin{equation}\label{def: p^T}
 \frac{d\prob^T}{d\prob} = \frac{\pare{X^{0,T}_T}^p}{\expec^{\prob}\bra{\pare{X^{0,T}_T}^p}}.
\end{equation}
The above densities are well-defined and strictly positive (cf. Assumption \ref{ass: wellpose} and Remark~\ref{rem: long run measure}), and $\prob^T=\prob$ in the logarithmic case $p=0$.
These myopic probabilities are interpreted as follows: an investor with relative risk aversion $1-p$ under the probability $\prob$ selects the same optimal payoff as another investor under the probability $\prob^T$, but with logarithmic utility, that is with unit risk aversion\footnote{These probabilities already appear in the work of \cite{Kramkov-Sirbu-06a,Kramkov-Sirbu-06b,Kramkov-Sirbu-07} under the name of $\bf R$.}. 

With the above definitions, the abstract version of the turnpike theorem reads as follows:
\begin{thm}[Abstract Turnpike]\label{thm: opt-port conv}
Let Assumptions \ref{ass: utility}--\ref{ass: wellpose} hold.
Then, for any $\epsilon>0$,
 \begin{enumerate}
  \item[a)] $\lim_{T\to \infty} \prob^T \pare{\sup_{u\in [0,T]} \left|r^T_u - 1\right| \geq \epsilon} =0$,
  \item[b)] $\lim_{T\to \infty} \prob^T \pare{\bra{\Pi^T, \Pi^T}_T \geq \epsilon} =0$, where $\bra{\cdot,\cdot}$ denotes the square bracket of semimartingales.
 \end{enumerate}
\end{thm}

\begin{rem}\label{rem: finite time}
\noindent
\begin{enumerate}
\item[i)]
Since $\prob^T\equiv \prob$ for $p=0$, convergence holds under $\prob$ in the case of logarithmic utility. In particular, the convergence holds on the entire time horizon $[0,T]$. Contrast this to the turnpike results for $p\neq 0$, in which convergence holds on a  time window $[0, t]$ for some fixed $t>0$.

\item[ii)]
Consider a market with the discounted asset prices
\[
 \frac{d \widetilde{S}^j_u}{\widetilde{S}^j_u} = \mu_u^j du + \sum_{k=1}^n \sigma_u^{jk} dW^k_u, \quad j=1, \cdots, d,
\]
where $\mu_u = \in \Real^d$, $\sigma_u \in \Real^{d\times n}$ for $t\geq 0$ and $W=(W^1, \cdots, W^n)'$ is a $\Real^n$-valued Brownian motion. The discounted optimal wealth processes satisfy
\[
 d \widetilde{X}^{i,T}_u = \widetilde{X}^{i,T}_u (\pi^{i,T}_u)' (\mu_u du + \sigma_u dW_u), \quad i=0,1,
\]
where $(\pi^{i,T})^{1\leq j\leq d}_{u\geq 0}$ represents the proportions of wealth invested in each risky asset.
In this case,  $[\Pi^T, \Pi^T]$ measures the square distance, weighted by
$\Sigma = \sigma\sigma'$, between the portfolios $\pi^{1,T}$ and $\pi^{0,T}$:
 \[
  \bra{\Pi^T, \Pi^T}_\cdot = \int_0^{\cdot} \pare{\pi^{1,T}_u-\pi^{0,T}_u}'\Sigma_u\pare{\pi^{1,T}_u-\pi^{0,T}_u}du.
 \]

\item[iii)]
The theorem implies that both optimal wealth processes and portfolios are close in any time window $[0, t]$ for any fixed $t>0$, under the probability $\prob^T$. Indeed, for any $\epsilon,t>0$:
\[
\lim_{T\to \infty} \prob^T\pare{\sup_{u\in [0,t]} \left|r^T_u - 1\right| \geq \epsilon} =0
\qquad\text{and}\qquad
\lim_{T\to \infty} \prob^T \pare{\bra{\Pi^T, \Pi^T}_t \geq \epsilon} =0.
\]
\end{enumerate}
\end{rem}

Except for logarithmic utility, Theorem \ref{thm: opt-port conv} is not a classic turnpike theorem, in that convergence holds under the probabilities $\prob^T$, which change with the horizon $T$.
However, since the events $\{\sup_{u\in[0,t]} |r^T_u-1| \geq \epsilon\}$ and
$\{[\Pi^T, \Pi^T]_t \geq \epsilon\}$ are $\F_t$-measurable, and any such event
$A$ satisfies $\prob^T(A)=\expec^\prob\bra{\indic_{A}\ d\prob^T/d\prob\big|_{\F_t}}$,  the relation between $\prob^T(A)$ and $\prob(A)$ depends on the (projected) density:
\begin{equation}\label{def: cond density}
\left.\frac{d\prob^T}{d\prob}\right|_{\F_t} = \frac{\expec^{\prob}_t\bra{\pare{X^{0,T}_T}^p}}{\expec^{\prob}\bra{\pare{X^{0,T}_T}^p}}.
\end{equation}
Understanding the convergence of these densities is the crucial step to bridge the gap from the abstract to the classic version of the turnpike theorem.

In fact, the densities in \eqref{def: cond density} become trivial under two additional assumptions: that the optimal CRRA strategy is myopic, and that its wealth process has independent returns. Under these assumptions, which hold in all the turnpike literature, with the exception of \cite[Theorem 1]{dybvig1999portfolio}, the density ${d\prob^T}/{d\prob}\big|_{\F_t}$ is independent of $T$, and the classic turnpike theorem follows:

\begin{cor}[IID Myopic Turnpike]\label{thm: power mayopic}
If, in addition to Assumptions \ref{ass: utility} --- \ref{ass: wellpose},
\begin{enumerate}
\item
$X^{0,T}_t=X^{0,S}_t\equiv X_t$ $\prob$-a.s. for all $t\le S,T$ (myopic optimality);

\item
$X_s/X_t$ and $\F_t$ are independent under $\prob$ for all $t\le s$ (independent returns).
\end{enumerate}
then, for any $\epsilon, t>0$:
 \begin{enumerate}
  \item[a)] $\lim_{T\to \infty} \prob \pare{\sup_{u\in [0,t]} \left|r^T_u - 1\right| \geq \epsilon} =0$,
  \item[b)] $\lim_{T\to \infty} \prob \pare{\bra{\Pi^T, \Pi^T}_t \geq \epsilon} =0$.
 \end{enumerate}
\end{cor}

In practice, if asset prices have independent returns, the optimal strategy for a CRRA investor entails a myopic portfolio with independent returns, and both conditions above hold. This is the case, for example, if asset prices follow exponential L\'{e}vy processes, as in \cite{Kallsen}.
Note however, that a myopic CRRA portfolio is not sufficient to ensure that $\prob^T$ is independent of $T$ (cf. Example~\ref{exa: zero cor} below).

Thus, the abstract turnpike readily yields a classic turnpike theorem under additional assumptions in Corollary \ref{thm: power mayopic}. However, even though these assumptions are common in the literature, they exclude  models in which portfolio choice is least tractable, and turnpike results are needed the most.
The next section proves classical and explicit turnpikes for diffusion models in which returns may not be independent, and the market may be incomplete.

\subsection{A Turnpike for Diffusions}\label{subsec: turnpike for diffusion}
This subsection states the classic turnpike theorem for a class of diffusion
models, in which a single state variable drives investment opportunities.
The state variable takes values in some interval $E=(\alpha,
\beta)\subseteq\Real$, with $-\infty \leq \alpha < \beta \leq \infty$, and
has the dynamics
\begin{equation}\label{eq: state sde}
 dY_t = b(Y_t) \, dt + a(Y_t) \, dW_t.
\end{equation}
The market includes a safe rate $r(Y_t)$ and $d$ risky assets, with prices $S^i_t$ satisfying
\[
 \frac{dS_t^i}{S_t^i} = r(Y_t)\, dt + dR_t^i, \quad 1\leq i \leq d,
\]
where the cumulative excess return process $R = (R^1, \cdots, R^d)'$ is defined as
\begin{equation}\label{eq: excess return}
 dR_t^i = \mu_i(Y_t) \, dt + \sum_{j=1}^d\sigma_{ij}(Y_t)\, dZ^j_t, \quad 1\leq i \leq d.
\end{equation}
Here $W$ and $Z = (Z^1, \cdots, Z^d)'$ are Brownian motions with correlations $\rho = (\rho^1, \cdots, \rho^d)'$, i.e. $d\langle
Z^i, W \rangle_t =  \rho^i(Y_t) \,dt$ for $1\leq i \leq d$. (The prime sign
is for matrix transposition.)

Denote by $\Sigma = \sigma \sigma'$, $A=a^2$, and $\Upsilon = \sigma \rho
a$. The first assumption on the model's coefficients concerns regularity and
non-degeneracy. Recall that for $\gamma \in (0,1]$ and an integer $k$, a function $f:E\mapsto\Real$ is \textsl{locally $C^{k,\gamma}$ on $E$} if for all bounded, open, connected $D\subset E$ such that $\bar{D}\subset E$ it follows that $f\in C^{k,\gamma}(\bar{D})$ (see \citet*[Chapter 5.1]{MR1625845} for a definition of the H\"{o}lder space $C^{k,\gamma}$). For integers $n, m$,
$C^{k,\gamma}(E, \Real^{n\times m})$ is the set of all $n\times m$ matrix-valued $f$ for which each component $f_{ij}$
is locally $C^{k,\gamma}$ on $E$. Write $\Real = \Real^{1\times 1}$ and
$\Real^n = \Real^{n\times 1}$. With this notation, assume:
\begin{ass}\label{ass: regular coeffs}
 $r\in C^{\gamma}(E, \Real)$, $b\in C^{1,\gamma}(E, \Real)$, $\mu \in C^{1,
   \gamma}(E, \Real^d)$, $A\in C^{2, \gamma}(E, \Real)$, $\Sigma\in C^{2,
   \gamma}(E, \Real^{d\times d})$, and $\Upsilon \in C^{2,\gamma}(E,
 \Real^d)$ for some $\gamma \in (0,1]$. For all $y\in E$, $\Sigma$ is positive
 definite and $A$ is positive.
\end{ass}

These regularity conditions imply the local existence and uniqueness of a
solution to the joint dynamics of the state variable and asset prices. The next assumption ensures the existence of a unique global solution, by requiring that {Feller's test for explosions} is negative \citep[Theorem 5.1.5]{Pinsky}.
\begin{ass}\label{ass: WP P}
There is some $y_0\in E$ such that
\begin{equation*}\label{ass: no_expl}
\int_{\alpha}^{y_0} \frac{1}{ A(y)m(y)}\left(\int_y^{y_0}m(z)dz\right) dy =
\infty =
\int_{y_0}^{\beta}\frac{1}{A(y)m(y)}\left(\int_{y_0}^{y}m(z)dz\right) dy,
\end{equation*}
where the speed measure is defined as $m(y):= \frac{1}{A(y)} \exp\pare{\int_{y_0}^y \frac{2b(z)}{A(z)} dz}$.
\end{ass}
Assumption \ref{ass: WP P} implies the model for $(R,Y)$ is well posed in that
it admits a solution. This statement is made precise within the setting of the
martingale problem, now introduced along with some notation.
Let $\Omega$ be the space of continuous maps $\omega: \Real_+ \to \Real^{n}$ and
$(\B_t)_{t\geq 0}$ be the filtration generated by the coordinate process $\Xi$
defined by $\Xi_t(\omega) = \omega_t$ for $\omega\in\Omega$. Let $\F =
\sigma\left(\Xi_t, t\geq 0\right)$ and $\F_t =\B_{t+}$.
For an open, connected set $D\subset\Real^n$ and $\gamma \in (0,1]$, let $\widetilde{A}\in C^{2,\gamma}(D,\Real^{n\times n})$ be point-wise positive definite and let $\widetilde{b}\in C^{1,\gamma}(D,\Real^{n})$. Define the second order elliptic operator
$\widetilde{L}$ by
\begin{equation*}\label{eq: gen op def}
\widetilde{L} =
\frac{1}{2}\sum_{i,j=1}^{n}\widetilde{A}_{ij}\frac{\partial^2}{\partial
  x_i\partial x_j}
+ \sum_{i=1}^{n}\widetilde{b}_i\frac{\partial}{\partial x_i}.
\end{equation*}
\begin{defn}\label{def: martingale problem}
A family of probability measures $(\prob^{x})_{x\in D}$ on $(\Omega,\F)$ is a solution to the martingale problem for $\widetilde{L}$ on
$D$ if, for each $x\in D$: $i): \prob^x(\Xi_0=x) =1$, $ii): \prob^x(\Xi_t \in D, \forall t\geq 0) =1$, and $iii): \left(f(\Xi_t) -f(\Xi_0) - \int_0^t
\widetilde{L}f(\Xi_u)\, du; (\F_t)_{t\geq 0}\right)$ is a $\prob^x$ martingale for all $f\in C^2_0(D)$.
\end{defn}
Let $\xi = (z,y)\in\Real^d\times E$. Consider the generator
\begin{equation}\label{eq: L def}
 L=\frac12 \sum_{i,j=1}^{d+1} \widetilde{A}_{ij}(\xi) \frac{\partial^2}{\partial
   \xi_i \partial \xi_j} + \sum_{i=1}^{d+1} \widetilde{b}_i(\xi)
 \frac{\partial}{\partial \xi_i}, \quad \widetilde{A}
 = \pare{\begin{array}{ll}1_d & 0 \\ 0 & A\end{array}}
 \text{ and } \widetilde{b} = \pare{\begin{array}{l}0\\ b\end{array}}.
\end{equation}
This is the infinitesimal generator of $(B,Y)$ from \eqref{eq: state sde} and
\eqref{eq: excess return} where $B$ is a $d$-dimensional Brownian Motion
starting at $z$, independent of $Y$ which starts at $y$.  Assumptions
\ref{ass: regular coeffs} and \ref{ass: WP P} imply the following:

\begin{prop}\label{lem: mart prob} Let Assumptions \ref{ass: regular coeffs} and \ref{ass: WP P} hold. Then there exists a unique solution $(\prob^{\xi})_{\xi\in\Real^{d}\times E}$ to the martingale problem on $\Real^{d}\times E$ for $L$ in \eqref{eq: L def}.
\end{prop}
\begin{rem}\label{rem: mart prob} There is a one to one correspondence between
  solutions to the martingale problem and weak solutions for $(B,Y)$, see
  \cite[Chapter V]{MR1780932}.  Since $A(y) > 0$ for $y\in E$, \emph{defining} $W$ via $W_t = \int_0^ta(Y_s)^{-1}\left(dY_s - b(Y_s)ds\right)$, $Z = \rho W + \bar{\rho}B$ where $\bar{\rho}$ is a square root of $1-\rho\rho'$, and $R$ via \eqref{eq: excess return}, it follows that $\pare{(R,Y), (W,B), (\Omega,\F,(\F_t)_{t\geq 0},\prob^{\xi})}$ is a weak solution of \eqref{eq: state sde} and \eqref{eq: excess return}.
\end{rem}

Assumption \ref{ass: WP P} is merely technical, in that it requires that the
original market is well defined. By contrast, the next assumption places some
restrictions on market dynamics.

\begin{ass}\label{ass: correl}
$\rho' \rho$ is constant (i.e. it does not depend on $y$), and $\sup_{y\in E} c(y) <\infty$, where
\begin{equation}\label{eq: potential def}
c(y) :=\frac{1}{\delta}(p r - \frac{q}{2} \mu'\Sigma^{-1} \mu)(y),\qquad y\in E,
\end{equation}
with $q:= p/(p-1)$ and $\delta:= (1-q \rho' \rho)^{-1}$.
\end{ass}

Assumption \ref{ass: correl} is straightforward to check, and holds when
$p\leq 0$ for virtually all models in the literature, with the exception of
correlation risk (cf. \citet*{buraschi2010correlation}).

Set $(\F^{R,Y}_t)_{t\geq 0}$ as the right continuous envelope of the filtration generated by $(R,Y)$. For any admissible strategy $H$ with respect to this filtration, the corresponding risky weight $\pi = H S/X$ is an adapted, $R$-integrable process $(\pi_t)^{1\le i\le d}_{t\ge 0}$, and satisfies the relation
\begin{equation}\label{eq: wealth dynamics}
\frac{d X^\pi_t}{X^\pi_t} = r(Y_t) dt+\pi'_t\, dR_t.
\end{equation}
In this Markovian setting, the value function for the horizon $T\in \Real_+$ is given by:
\begin{equation}\label{def: u}
 u^{0,T} =  u^T(t,x,y) = \sup_{\pi \text{ admissible}} \expec^{\prob}\bra{\pare{X_{T}^{\pi}}^p/p \such X_t=x, Y_t=y}, \quad \text{ for } t\in [0,T].
\end{equation}
These utility maximization problems are studied at all horizons under the following assumption:
\begin{ass}\label{ass: ptphat}
There exist $(\hat{v}, \lambda_c)$ such that $\hat v\in C^2(E)$, $\hat v>0$, and solves the equation:
\label{ass: WP hat-v}
\begin{equation}\label{eq: eign eqn}
\cL \,v + c\, v = \lambda \, v, \quad y \in E,
\end{equation}
where
\begin{equation}\label{eq: cL B def}
\cL:= \frac12 A \,\partial^2_{yy} + B \,\partial_y;\qquad B:=
b-q\Upsilon' \Sigma^{-1} \mu.
\end{equation}
Also, for the $y_0\in E$ in Assumption \ref{ass: WP P}\footnote{Any $y_0\in
  E$ suffices. This $y_0$ is chosen to align $m$ with $\hat{m}$.}
\begin{align}\label{ass: recurrent}
\int_{\alpha}^{y_0} \frac{1}{\hat{v}^2A\hat{m}(y)} dy =& \infty
&
\int_{y_0}^{\beta}\frac{1}{\hat{v}^2 A \hat{m} (y)} dy =& \infty,\\
\label{ass: posrec l1}
\int_{\alpha}^{\beta} \hat{v}^2\, \hat{m} (y) \, dy  =& 1
&
\int_{\alpha}^{\beta} \hat{v}\, \hat{m}(y) \,dy <& \infty,
\end{align}
where
\begin{equation}\label{eq: mnu def}
\hat{m}(y):= \frac{1}{A(y)} \exp\pare{\int_{y_0}^y \frac{2B(z)}{A(z)} dz}.
\end{equation}
\end{ass}
\begin{rem}
If $\hat{v} > 0$ satisfies \eqref{eq: eign eqn}, \eqref{ass: recurrent}, the inequality in \eqref{ass: posrec l1}, then $\int_{\alpha}^{\beta}\hat{v}^2\hat{m}(y)dy = 1$ is equivalent to $\int_{\alpha}^{\beta}\hat{v}^2\hat{m}(y)dy < \infty$, up to a renormalization of $\hat{v}$. We assume that the integral equals one only for convenience of notation.
\end{rem}

Assumption \ref{ass: ptphat} is interpreted as follows. Equation \eqref{eq: eign eqn} is the ergodic HJB equation, which controls the long-run limit of
the utility maximization problem (cf. \citet*{Guasoni-Robertson} Theorem~7 and
Section~2.2.1). Its solution $\hat v$ is related to the finite-horizon value
functions $u^T$ by $u^T(x,y,0) \sim (x^p/p) (e^{\lambda T} \hat v(y))^\delta$.
Thus, assuming that \eqref{eq: eign eqn} has a solution guarantees that the
long-run optimization problem is well posed. The presence of $\delta$
reflects the power transformation of \cite{Zariphopoulou}, which allows to
write the ergodic HJB equation in the linear form \eqref{eq: eign eqn}.

To understand the meaning of \eqref{ass: recurrent} and \eqref{ass: posrec l1}, define the operator:
\begin{equation}\label{eq: phat L def}
 \hat{\cL}=\frac12 \sum_{i,j=1}^{d+1} \widetilde{A}_{ij}(\xi) \frac{\partial^2}{\partial
   \xi_i \partial \xi_j} + \sum_{i=1}^{d+1} \hat{b}_i(\xi)
 \frac{\partial}{\partial \xi_i}, \quad \hat{b}
 = \pare{\begin{array}{l}-q\bar{\rho}'\sigma'\Sigma^{-1}\left(\mu + \delta\Upsilon\frac{\hat{v}_y}{\hat{v}}\right)\\ B+A\frac{\hat{v}_y}{\hat{v}}\end{array}},
\end{equation}
where $\widetilde{A}$ is the same as in \eqref{eq: L def}.
Condition \eqref{ass: recurrent} in Assumption~\ref{ass: ptphat} implies that
the martingale problem for $\hat{\cL}$ on $\Real^d \times E$ has a unique
solution $(\hat{\prob}^\xi)_{\xi\in \Real^d \times E}$ and that
$\hat{\prob}^\xi$ is equivalent to $\prob^\xi$ (see Lemma~\ref{lemma: WP hat-P} below). The family $(\hat{\prob}^{\xi})_{\xi\in\Real^{d}\times E}$ is
called the \emph{long-run probability}. Girsanov's theorem in turn
implies that the following stochastic differential equation has a unique weak solution starting at $\xi$ under $\hat{\prob}^{\xi}$:
\begin{equation}\label{eq: phat sde}
\begin{split}
 & dR_t = \frac{1}{1-p} \pare{\mu+ \delta \Upsilon \frac{\hat{v}_y}{\hat{v}}}(Y_t) \, dt + \sigma(Y_t)\, d \hat{Z}_t,\\
 & dY_t = \pare{B+ A\frac{\hat{v}_y}{\hat{v}}}(Y_t) \, dt + a(Y_t)\, d\hat{W}_t,
\end{split}
\end{equation}
Here, $(\hat{B},\hat{W})$ is a $d+1$ dimensional Brownian
Motions under $\hat{\prob}^{\xi}$, and $\hat{Z} = \rho \hat{B} + \overline{\rho} \hat{W}$.  Conditions \eqref{ass: recurrent} and \eqref{ass: posrec l1} imply that $Y$ is ergodic under $(\hat{\prob}^{\xi})_{\xi\in\Real^{d}\times E}$ (Lemma \ref{lemma: hat-P posrec} below, and Section \ref{SS: phat Y prop} for a precise definition of ergodicity). This property drives the long-run asymptotics of the projected densities in \eqref{def: cond density}.

A simple criterion to check Assumption \ref{ass: ptphat} is the following:
\begin{prop}\label{prop: turnpike holds}
Let Assumptions \ref{ass: regular coeffs} and \ref{ass: correl} hold. If $c$ in \eqref{eq: potential def} and $\hat{m}$ in \eqref{eq: mnu def} satisfy:
\begin{align}
\label{eq: m prob requirement}
& \int_{\alpha}^{\beta} \hat{m}(y) dy < \infty,\\
\label{eq: c decay requirement}
& \lim_{y\downarrow\alpha} c(y) = \lim_{y\uparrow\beta} c(y) =
  -\infty.
\end{align}
Then, Assumption \ref{ass: ptphat} holds.
\end{prop}

\begin{rem}
If the interest rate $r$ is bounded from below, and $p< 0$, \eqref{eq: c decay requirement} states that the squared norm of the vector of risk premia $\sigma^{-1}\mu$ goes to $\infty$ at the boundary of the state space $E$.
\end{rem}
Assumption \ref{ass: ptphat} guarantees that (see Proposition~\ref{lemma: HJB verification} below) at all finite horizons $T$, the value functions $u^T$ in \eqref{def: u} can be represented as $u^T(t, x, y) = (x^p/p) (v^T(t,y))^\delta$ for $(t,x,y)\in [0,T]\times \Real_+ \times E$, where $v^T$ is a strictly positive classical solution to the linear parabolic PDE:
\begin{equation}\label{eq: HJB}
 \begin{split}
   & \partial_t v + \cL v + c\, v = 0, \quad (t,y) \in (0,T) \times E,\\
   & v(T,y) = 1, \hspace{1.7cm} y \in E.
 \end{split}
\end{equation}
Moreover, the optimal portfolio for the horizon $T$ is (all functions are evaluated at $(t,Y_t)$):
\begin{equation}\label{eq: finite opt}
   \pi^T = \frac{1}{1-p} \Sigma^{-1}\pare{\mu + \delta \Upsilon \frac{v^T_y}{v^T}}.
\end{equation}
Thus, the wealth process corresponding to this portfolio leads to the optimal terminal wealth $X^{\pi^T}_T$, which in turn defines the probability\footnote{Since $R_0=0$ by assumption, $\prob^\xi$ with $\xi = (0,y)$ is denoted as $\prob^y$. The same convention applies to $\hat{\prob}^\xi$.} $\prob^{T,y}$ by \eqref{def: p^T}.
Understanding the convergence of $d\prob^{T,y}/d\prob^y|_{\F_t}$ is key to go beyond the abstract version of the turnpike. To this end, observe from \eqref{eq: phat sde} that the portfolio:
\begin{equation}\label{eq: long-run opt}
   \hat{\pi} = \frac{1}{1-p} \Sigma^{-1}\pare{\mu + \delta \Upsilon \frac{\hat{v}_y}{\hat{v}}}
\end{equation}
evaluated at $(t, Y_t)$, is optimal for logarithmic utility under the probability $\hat\prob^y$. This fact suggests that the conditional densities of $\hat\prob^y$ are natural candidates for the limits of the conditional densities
$d\prob^{T,y}/d\prob^y|_{\F_t}$. Combined with the ergodicity of $Y$ under $(\hat{\prob}^y)_{y\in E}$ the next
result follows:

\begin{lem}\label{thm: conv cond den}
Let Assumptions~\ref{ass: regular coeffs}, \ref{ass: WP P}, \ref{ass: correl},
and \ref{ass: ptphat} hold. Then, for all $y\in E$ and $t, \eps > 0$:
\begin{equation}\label{eq: conv cond den}
\lim_{T\rightarrow \infty}
\hat{\prob}^y\pare{\left|\frac{d\prob^{T,y}}{d\hat{\prob}^y}\bigg|_{\F_t} -1 \right| \geq \epsilon} = 0.
\end{equation}
\end{lem}

This result essentially allows to replace $\prob^{T,y}$ in Theorem \ref{thm: opt-port conv} with $\hat\prob^y$. Then, the classic turnpike follows from the equivalence of $\hat\prob^y$ and $\prob^y$ (cf. Lemma~\ref{lemma: WP hat-P}, part (ii)):

\begin{thm}[Classic Turnpike]\label{thm: power 1-factor}
Let Assumptions \ref{ass: utility} - \ref{ass: wellpose}, \ref{ass: regular coeffs}, \ref{ass: WP P}, \ref{ass: correl} and \ref{ass: ptphat} hold.
Then, for all $y\in E$,  $0\ne p<1$ and  $\epsilon,t>0$:
\begin{enumerate}[a)]
\item $\lim_{T\to \infty} \prob^y \,(\sup_{u\in [0,t]} \left|r^T_u - 1\right| \geq \epsilon) =0$,
\item $\lim_{T\to \infty} \prob^y \pare{\bra{\Pi^T, \Pi^T}_t \geq \epsilon} =0$.
\end{enumerate}
\end{thm}
Abstract and classic turnpikes compare the finite-horizon optimal portfolio of
a generic utility to that of its CRRA benchmark at the same finite horizon. By
contrast, the explicit turnpike, discussed next, uses as a benchmark the long horizon limit of the optimal CRRA portfolio.

This result has two main implications: first, and most importantly, it shows
that the two approximations of replacing a generic utility with a power, and a finite horizon problem with its long-run limit, lead to small errors as the horizon becomes large. Second, this theorem has a nontrivial statement even for $U$ in the CRRA class: in this case, it states that the optimal finite-horizon portfolio converges to the long-run optimal portfolio, identified as a solution to the ergodic HJB equation \eqref{eq: eign eqn}.

To state the explicit turnpike, define, in analogy to \ref{def: ratio log}, the ratio of optimal wealth processes relative to the long-run benchmark, and their stochastic logarithms as:
\[
\hat{r}^T_{u} := \frac{X^{1,T}_u}{\hat{X}_u}, \qquad
\hat{\Pi}^T _u:= \int_0^{u} \frac{d\hat{r}^T_v}{\hat{r}^T_{v-}}, \quad \text{ for } u \in [0,T],
\]
where $\hat{X}$ is the wealth process of the long-run portfolio
$\hat{\pi}$.
\begin{thm}[Explicit Turnpike]\label{cor: power 1-factor myopic}
Under the assumptions of Theorem~\ref{thm: power 1-factor}, for any $y\in E$, $\epsilon,t>0$ and $0\ne p<1$:
\begin{enumerate}[a)]
\item $\lim_{T\to \infty} \prob^y \,(\sup_{u\in [0,t]} \left|\hat{r}^T_u - 1\right| \geq \epsilon) =0$,
\item $\lim_{T\to \infty} \prob^y \pare{\bra{\hat{\Pi}^T, \hat{\Pi}^T}_t \geq \epsilon} =0$.
\end{enumerate}
If $U$ is in CRRA class, \eqref{ass: growth} is not needed for the above convergence.
\end{thm}

\subsection{Applications}
Before proving the main results, we offer two examples of their significance.
We begin with an application to target-date mutual funds and the social planner problem.
\begin{exa}
Consider several investors, who differ in their initial capitals $(x_i)_{i=1}^n$ and risk aversions $(\gamma_i)_{i=1}^n$, but share the same long horizon $T$. Suppose that they do not invest independently, but rather pool their wealth into a common fund, delegate a manager to invest it, and then collect the proceeds on their respective capitals under the common investment strategy. This setting is typical of target-date retirement funds, in which savings from a diverse pool of participants are managed according to a single strategy, characterized by the common horizon $T$.

Suppose the manager invests as to maximize a weighted sum of the investors' expected utilities, thereby solving the problem
\[
\max_{X \in \X^T} \sum_{i=1}^n w_i\expec^{\prob}
\bra{\frac{(x_i X_T)^{1-\gamma_i}}{1-\gamma_i}}
\]
for some $(w_i)_{i=1}^n >0$.
By homogeneity and linearity, this problem is equivalent to maximizing the expected value $\expec^{\prob}[U(X_T)]$ of the master utility function\footnote{If a logarithmic investor is present ($\gamma_i=1$ for some $i$), a constant is added to $U(x)$, and the stated equivalence remains valid.}:
\begin{equation}
U(x)=\sum_{i=1}^n \tilde w_i
\frac{x^{1-\gamma_i}}{1-\gamma_i}
\qquad\text{where }\tilde w_i =w_i {x_i^{1-\gamma_i}} .
\end{equation}
Thus, the fund manager is akin to a social planner, who ponders the welfare of various investors according to the weights $\tilde w_i$. The question is how these weights affect the choice of the common fund's strategy, if the horizon is distant, as for most retirement funds.

While this problem is intractable for a fixed horizon $T$, turnpike theorems offer a crisp solution in the long run limit. Indeed, the master utility function satisfies Assumption \ref{ass: utility} with $\gamma=1-p = \min_{1\le i\le n}\gamma_i$. Thus, for any market that satisfies the additional Assumptions
\ref{ass: safe asset}--\ref{ass: wellpose}, it is optimal for the fund manager to act on behalf of the least risk-averse investor.

The implication is that most or nearly all fund participants will find that the fund takes more risk than they would like, \emph{regardless} of the welfare weights $\tilde w_i$ (provided that they are strictly positive).
The result holds irrespective of market completeness or independence of returns, and indicates that a social planner objective is ineffective in choosing a portfolio that balances the needs or investors with different preferences.

Note that this result points in the same direction as the ones of \citet*{benninga2000hao} and \citet*{cvitanic:apf}, with the crucial difference that prices are endogenous in their models, while they are exogenous in our setting. Finally, the result should be seen in conjunction with the classical numeraire property of the log-optimal portfolio, whereby the wealth process of the logarithmic investor becomes arbitrarily larger than any other wealth process. In spite of this property, the fund manager does not choose the log-optimal strategy, but the one optimal for the least risk-averse investor.
\end{exa}

The next example is more technical. In the model that follows, returns of risky assets have constant volatility, but their drift is an independent Ornstein-Uhlenbeck process. Such independence entails that optimal CRRA portfolios are myopic, as in assumption $i)$ of Corollary~\ref{thm: power mayopic}. By contrast, the time-varying drift makes returns dependent over time, hence assumption $ii)$ of the same corollary does not hold.
As a result, the proof of Corollary~\ref{thm: power mayopic} fails, because it requires that $d\prob^{T,y}/d\prob^y\big|_{\F_t}$ is constant with respect to $T$, which is not the case here. Yet, both the classic and the explicit turnpikes hold in this model, in the form of Theorems \ref{thm: power 1-factor} and \ref{cor: power 1-factor myopic}. Of course, these results depend on the ergodicity of the diffusion $Y$.

\begin{exa}\label{exa: zero cor}
Consider the diffusion model:
\[dR_t = Y_t \,dt + d Z_t \quad  \text{ and } \quad dY_t = -Y_t \,dt + dW_t,\]
where the correlation between $Z$ and $W$ is $\rho=0$ and the safe rate
$r>0$. Clearly, Assumptions \ref{ass: regular coeffs} and \ref{ass: WP P}
hold. Furthermore, for $p < 0$ Assumption \ref{ass: correl}, as well as the hypotheses of
Proposition \ref{prop: turnpike holds} are met and hence Assumption \ref{ass: ptphat} holds as
well, yielding the results of Theorems \ref{thm: power 1-factor} and \ref{cor: power 1-factor myopic}. The optimal portfolio
for a CRRA investor is a myopic portfolio $\pi^T_t= Y_t/(1-p)$; see
\eqref{eq: finite opt}. However, the conditional density
$d\prob^{T,y}/d\prob^y\big|_{\F_t}$ depends on the horizon $T$. Indeed, it follows from Proposition \ref{prop: Ito-Girsanov result} below that
\[
 \frac{d\prob^{T,y}}{d\prob^y}|_{\F_t} = \Exp \pare{\int \frac{v^T_y(s, Y_s)}{v^T(s, Y_s)}\, dW_s - q \int_0^\cdot Y_s \, dZ_s}_t,
\]
where $v^T$ satisfies the HJB equation $\partial_t v + \frac12 \partial^2_{yy}
v- y \partial_y v + (rp-\frac{q}2 y^2) v=0$ with $v(T,y)=1$.
The above conditional density is independent of $T$ only if
$g^T(t,y):=v^T_y(t,y)/v^T(t,y)$ is independent of $T$ for any fixed
$(t,y)$. It can be shown
that $v^T$ is smooth, and not just twice continuously differentiable, in the
state variable $y$, and hence $g^T$ satisfies $\partial_t g + \frac12 \partial^2_{yy} g + (g-y)\partial_y g -g - qy=0$ with $g(T, y)=0$. If $g^T$ was independent of $T$, $0$ should be a solution to the previous equation. However, this is clearly not the case for $q\neq 0$.
\end{exa}

\section{Proof of the Abstract Turnpike}
This section contains the results leading to the abstract version of the turnpike theorem. The proof proceeds through two main steps:
\begin{enumerate}
\item[i)]
Establish that optimal payoffs for the generic utility converge to their CRRA counterparts;

\item[ii)]
Obtain from the convergence of optimal payoffs the convergence of wealth processes.
\end{enumerate}

\subsection{Convergence of optimal payoffs}\label{SS : conv opt payoff}
First, note that Assumption \ref{ass: NFLVR} implies the existence of a \emph{deflator}, that is a strictly positive process $Y$ such that $Y X$ is a (nonnegative) supermartingale for all $X\in \mathcal{X}^T$ and $T>0$. Condition \eqref{ass: growth} entails that $\lim_{T\rightarrow \infty}\expec[Y_T]=0$ for any such deflator $Y$. In this section, the capital letter $Y$ is used for deflators, while in the section on diffusion models it denotes the state variable.
Recall a result from \cite{karatzas.zitkovic.03}:
\begin{thm}[Karatzas-\v{Z}itkovi\'{c}]\label{thm: foc}
 Under Assumptions \ref{ass: utility} - \ref{ass: wellpose}, the optimal payoffs are
 \begin{equation}\label{eq: terminal wealth}
  X^{i,T}_T = I^i(y^{i,T} Y^{i,T}_T), \quad i=0,1, \,T>0,
 \end{equation}
 where $I^0$ is the inverse function of $x^{p-1}$, $I^1$ is the inverse function of $U'(x)$, the positive constant $y^{i,T}$ is the Lagrangian multiplier, and $Y^{i,T}$ is some supermartingale deflator. Moreover,
 \begin{equation}\label{eq: first order cond}
  y^{i,T} = \expec^\prob\bra{(U^i)'(X^{i,T}_T) X^{i,T}_T} \geq \expec^\prob\bra{(U^i)'(X^{i,T}_T) X_T}, \quad i=0, 1, \,T>0,
 \end{equation}
 for any $X\in \mathcal{X}^T$. Here $U^0(x) = x^p/p$ and $U^1(x) = U(x)$.
\end{thm}

\begin{rem}\label{rem: long run measure}~
\begin{enumerate}
 \item[i)] It follows from \eqref{eq: terminal wealth} and the Inada condition that $X^{i,T}_T >0$ $\prob$-a.s. for $i=0, 1$ and $T\geq 0$. Since $X^{i,T}$ is a nonnegative $\qprob^T$-supermartingale and $\qprob^T$ is equivalent to $\prob$, it follows that $X^{i,T}_t >0$ $\prob$-a.s. for $0\leq t\leq T$.
 \item[ii)] Condition \eqref{ass: growth} entails that $\lim_{T\rightarrow \infty} \expec^\prob[Y^{i,T}_T] =0$ for $i=0, 1$ and $\lim_{T\rightarrow \infty} \prob^T (S_T^0 \geq N)=1$ for any $N>0$.
 \item[iii)] Recall the probability measure $\prob^T$ defined in \eqref{def: p^T}. The optimal wealth process $X^{0,T}$ has the num\'{e}raire property under $\prob^T$, i.e. $\expec^{\prob^T}[X_T / X^{0,T}_T] \leq 1$ for any $X\in \mathcal{X}^T$. This claim follows from $\expec^{\prob}\bra{(X^{0,T}_T)^p \pare{X_T/X^{0,T}_T -1}} \leq 0$, obtained from \eqref{eq: first order cond}, and switching the expectation from $\prob$ to $\prob^T$.
\end{enumerate}
\end{rem}

Both $X^{0,T}_T$ and $X^{1,T}_T$ will be shown to be unbounded as
$T\rightarrow \infty$. However, the main result of this subsection, Lemma
\ref{lem: r->1 power}, shows that their ratio at the horizon $T$, given by $r^T_T$
from \eqref{def: ratio log} satisfies $\lim_{T\rightarrow \infty}
\expec^{\prob^T} \bra{|r^T_T-1|} =0$. Lemma \ref{lem: r->1 power} will be the
culmination of a series of auxiliary results. Assumptions \ref{ass: utility} -
\ref{ass: wellpose} are enforced in the rest of this subsection.

\nada{To ease
the presentation, the following notational changes will be used
\emph{within the proofs} of all Lemma's within this subsection:
\[
\begin{split}
X_T &:= X^{1,T}_T;\quad Y_T := Y^{1,T}_T;\quad y_T = y^{1,T}\\
\tX &:= X^{0,T}_T;\quad \tY_T := Y^{0,T}_T;\quad \tsmY := y^{0,T} =
\expec^{\prob}\bra{\tX^p}\\
r_T&:= r^T_T;\quad \fR_T := \fR(X_T)
\end{split}
\]
}

\begin{lem}\label{lemma: opt unbounded power}
It holds that
\begin{equation*}
  \lim_{T\to \infty} \prob^T \pare{X^{0,T}_T\geq N} =1, \quad \text{ for any } N>0.
 \end{equation*}
\end{lem}
\begin{proof}
It suffices to prove $\limsup_{T\to \infty} \prob^T (X^{0,T}_T <N)=0$ for each fixed $N$. To this end, the num\'eraire property of $X^{0,T}$ under $\prob^T$ implies that:
 \[
  1\geq \expec^{\prob^T}\bra{\frac{S^0_T}{X^{0,T}_T}} \geq \expec^{\prob^T} \bra{\frac{S^0_T}{X^{0,T}_T} \, \indic_{\set{X^{0,T}_T<N, S^0_T\geq \tilde{N}}}} \geq \frac{\tilde{N}}{N}\, \prob^T\pare{X^{0,T}_T < N, S^0_T \geq \tilde{N}},
 \]
 for any positive constant $\tilde{N}$. As a result, $\prob^T(X^{0,T}_T<N, S^0_T \geq \tilde{N})\leq N / \tilde{N}$. Combining the last inequality with Remark \ref{rem: long run measure} item ii), it follows that
 \[
  \limsup_{T\to \infty} \prob^T (X^{0,T}_T<N) \leq \limsup_{T\to \infty} \prob^T (X^{0,T}_T<N, S^0_T \geq \tilde{N}) + \lim_{T\to \infty} \prob^T (S^0_T < \tilde{N}) \leq \frac{N}{\tilde{N}}.
 \]
Then, the statement follows since $\tilde{N}$ is chosen arbitrarily.
\end{proof}

Recall the Lagrangian multipliers $y^{i,T}$, $i=0, 1$,  from Theorem \ref{thm: foc}. The following result presents the asymptotic behavior of $y^{0,T} / y^{1,T}$ as $T\rightarrow \infty$.

\begin{lem}\label{lemma: ratio lagrangian}
 \[
  \liminf_{T\rightarrow \infty} \frac{y^{0,T}}{y^{1,T}} \geq 1.
 \]
\end{lem}
\begin{proof}
 The statement will be proved separately for $p=0$, $p\in(0,1)$, and $p<0$. Throughout this proof, in order to ease notation, we set $\alpha_T = y^{1,T}$, $Y_T = Y^{1,T}_T$, $\widetilde{Y}_T = Y^{0,T}_T$, $X_T = X^{1,T}_T$, $\widetilde{X}_T = X^{0,T}_T$, and $I = (U')^{-1}$. All expectations are under $\prob$. Observe first that
 \begin{equation}\label{eq: conv I}
   \lim_{y \downarrow 0} I(y)y^{\frac{1}{1-p}}=1.
 \end{equation}
 Indeed, set $x= I(y)$, hence $x\uparrow \infty$ as $y\downarrow 0$. Then the convergence above follows from \eqref{ass: conv} via
 \[
  \frac{I(y)}{y^{\frac{1}{p-1}}} = \frac{I(U'(x))}{(U'(x))^{\frac{1}{p-1}}} = \frac{x}{(U'(x))^{\frac{1}{p-1}}} = \pare{\frac{x^{p-1}}{U'(x)}}^{\frac{1}{p-1}} \rightarrow 1, \quad \text{ as } y\downarrow 0.
 \]

 \vspace{2mm}
 \noindent\underline{Case $p=0$:} It follows from \eqref{eq: conv I} that for any $\epsilon>0$ there exists $\delta>0$ such that $1-\epsilon \leq y I(y)\leq 1+\epsilon$ for $y<\delta$. Then \eqref{eq: terminal wealth} and \eqref{eq: first order cond} imply
 \[
 \begin{split}
  1 & = \expec[Y_T I(\alpha_T Y_T)] = \expec[Y_T I(\alpha_T Y_T) \indic_{\{\alpha_T Y_T < \delta\}} + Y_T I(\alpha_T Y_T) \indic_{\{\alpha_T Y_T \geq \delta\}}]\\
  &\leq \frac{1+\epsilon}{\alpha_T} \prob(\alpha_T Y_T <\delta) + I(\delta) \expec[Y_T\indic_{\{\alpha_T Y_T \geq \delta\}}]\\
  &\leq \frac{1+\epsilon}{\alpha_T} + I(\delta) \expec[Y_T],
 \end{split}
 \]
where the first inequality follows because $I$ is decreasing. Now, the previous inequality combined with Remark \ref{rem: long run measure} item ii) implies that
 \[
  1\leq \liminf_{T\rightarrow \infty} \frac{1+\epsilon}{\alpha_T},
 \]
 from which the statement follows since for $p=0$, $y^{0,T}=1$ and $\epsilon$ is chosen arbitrarily.

 \vspace{2mm}
 \noindent\underline{Case $p\in (0,1)$:} It follows from \eqref{ass: conv} that for any $\epsilon>0$ there exists $M>0$ such that $1-\epsilon \leq U'(x) x^{1-p} \leq 1+\epsilon$ for $x\geq M$. Then \eqref{eq: first order cond} implies that
 \[
 \begin{split}
  1 &= \frac{1}{\alpha_T}\expec\bra{ U'(X_T) X_T} = \frac{1}{\alpha_T}\expec\bra{ U'(X_T) X_T^{1-p} X^p_T \,\indic_{\{X_T\geq M\}}} + \frac{1}{\alpha_T}\expec\bra{ U'(X_T) X_T \,\indic_{\{X_T\leq M\}}}\\
  & \leq  \frac{1+\epsilon}{\alpha_T}\expec\bra{X^p_T \,\indic_{\{X_T\geq M\}}} +\frac{1}{\alpha_T}\expec\bra{ U'(X_T) X_T \,\indic_{\{X_T\leq M\}}}.
 \end{split}
 \]
 Note that $(1/\alpha_T)\expec\bra{ U'(X_T) X_T \,\indic_{\{X_T\leq M\}}} = \expec[Y_T X_T \,\indic_{\{X_T\leq M\}}]\leq M \expec[Y_T] \rightarrow 0$, as $T\rightarrow \infty$. Therefore
 \[
  \frac{1}{1+\epsilon} \leq \liminf_{T\rightarrow \infty} \frac{1}{\alpha_T} \expec\bra{X^p_T \,\indic_{\{X_T\geq M\}}} \leq \liminf_{T\rightarrow \infty} \frac{1}{\alpha_T} \expec[X^p_T] \leq \liminf_{T\rightarrow \infty} \frac{1}{\alpha_T} \expec[\widetilde{X}_T^p],
 \]
 where the third inequality follows from the optimality of $\widetilde{X} = X^{0,T}$ for $\sup_{x\in \mathcal{X}^T}\expec[X^p_T/p]$. Note that $y^{0,T} = \expec[\widetilde{X}_T^p]$. The statement follows from the previous inequality since $\epsilon$ is chosen arbitrarily.

 \vspace{2mm}
 \noindent\underline{Case $p<0$:} For any $\epsilon>0$ there exists $\delta>0$
 such that $1-\epsilon \leq I(y) y^{\frac{1}{1-p}} \leq 1+\epsilon$ for
 $y<\delta$. Then \eqref{eq: terminal wealth} and \eqref{eq: first order cond}
 yield (recall $q = p/(p-1)$ is the conjugate exponent to $p$)
 \[
 \begin{split}
  1& = \expec\bra{Y_T I(\alpha_T Y_T)} = \expec\bra{Y_T I(\alpha_T Y_T) \,\indic_{\{\alpha_T Y_T<\delta\}}} + \expec\bra{Y_T I(\alpha_T Y_T) \,\indic_{\{\alpha_T Y_T \geq \delta\}}}\\
  &\leq \frac{1+\epsilon}{\alpha_T^{\frac{1}{1-p}}} \expec\bra{Y^q_T \,\indic_{\{\alpha_T Y_T<\delta\}}} +  \expec\bra{Y_T I(\alpha_T Y_T) \,\indic_{\{\alpha_T Y_T \geq \delta\}}}.
 \end{split}
 \]
 Since $\expec\bra{Y_T I(\alpha_T Y_T) \,\indic_{\{\alpha_T Y_T \geq \delta\}}} \leq I(\delta) \expec[Y_T] \rightarrow 0$, as $T\rightarrow \infty$, the inequality in the last line yields
 \[
  \frac{1}{1+\epsilon} \leq \liminf_{T\rightarrow \infty} \frac{1}{\alpha^{\frac{1}{1-p}}_T} \expec\bra{Y^q_T \,\indic_{\{\alpha_T Y_T<\delta\}}} \leq \liminf_{T\rightarrow\infty} \frac{1}{\alpha_T^{\frac{1}{1-p}}} \expec\bra{Y^q_T}.
 \]
 The ${1-p}^{th}$ power on both sides of the previous inequality gives
\[
  \pare{\frac{1}{1+\epsilon}}^{1-p} \leq \liminf_{T\rightarrow \infty} \frac{1}{\alpha_T} \expec[Y^q_T]^{1-p} \leq \liminf_{T\rightarrow \infty} \frac{1}{\alpha_T} \expec[\widetilde{X}_T^p],
\]
 from which the statement follows.  Since $p<0$, the second inequality above follows from
 \[
 \frac{1}{p} \,\expec[\widetilde{X}_T^p] = \frac{1}{p}\,\expec[\widetilde{Y}_T^q]^{1-p} \leq \frac{1}{p} \, \expec[Y_T^q]^{1-p},
 \]
 where the equality holds due to the duality for power utility and the inequality follows from the optimality of $\widetilde{Y}$ for the dual problem  which minimizes $\expec[-\mathcal{Y}_T^q / q]$ among all supermartingale deflators $\mathcal{Y}$.
\end{proof}

The previous two lemmas combined describe the asymptotic behavior of $X^{1,T}_T$
and $\fR(X^{1,T}_T)$ where $\fR$ is given in \eqref{eq: marginal ratio function}.

\begin{lem}\label{lemma: asy RX1}
It holds that
 \[
  \lim_{T\rightarrow \infty} \prob^T(X_T^{1,T} \geq N) =1, \quad \text{ for any } N>0.
 \]
 Hence
 \[
  \lim_{T\rightarrow \infty} \prob^T(|\fR(X^{1,T}_T) -1| \geq \epsilon) =0, \quad \text{ for any } \epsilon>0.
 \]
\end{lem}

\begin{proof}
 It follows from Lemma \ref{lemma: ratio lagrangian} and \eqref{eq: first order cond} that
\[
  2\geq \frac{y^{1,T}}{y^{0,T}}  \geq \frac{\expec^\prob[X^{0,T}_T U'(X^{1,T}_T)]}{\expec^\prob[(X^{0,T}_T)^p]} = \expec^{\prob^T}\bra{\frac{U'(X^{1,T}_T)}{(X^{0,T}_T)^{p-1}}}, \quad \text{ for sufficiently large } T.
\]
 Combining the previous inequality with Lemma \ref{lemma: opt unbounded power}, the first statement follows.
 Indeed, for any given $M$ and $N$, on the set $\{X^{1,T}_T\leq N; X^{0,T}_T\geq M\}$, $(X^{0,T}_T)^{1-p} \geq M^{1-p}$ and $U'(X^{1,T}_T)\geq U'(N)$, therefore
\[
   2\geq \expec^{\prob^T} \bra{\frac{U'(X^{1,T}_T)}{(X^{0,T}_T)^{p-1}} \,\indic_{\{X^{1,T}_T\leq N; X^{0,T}_T\geq M\}}} \geq U'(N) M^{1-p} \,\prob^T(X^{1,T}_T\leq N; X^{0,T}_T\geq M).
\]
Hence,
\begin{equation*}
 \prob^T(X^{1,T}_T \leq N) \leq \prob^T(X^{1,T}_T\leq N; X^{0,T}_T\geq M) + \prob^T(X^{0,T}_T\leq M) \leq \frac{2}{U'(N)M^{1-p}} + \prob^T(X^{0,T}_T\leq M).
\end{equation*}
Letting first $T\rightarrow \infty$ and then $M\rightarrow \infty$ in the previous inequality, the first statement follows.

We move to the proof of the second statement. For any $\eps>0$, due to \eqref{ass: conv}, there exists a sufficiently large $N_{\eps}$ such that $|\fR(x)-1|<\eps$ for any $x>N_{\eps}$. As a result, $\prob^T(|\fR(X^{1,T}_T)-1|\geq \eps, X^{1,T}_T> N_{\eps}) = 0$. Combining the previous identity with $\lim_{T\to \infty} \prob^T(X^{1,T}_T\leq N_{\eps})=0$, the second statement follows.
\end{proof}

We continue with the following result, which is crucial for the proof of Lemma
\ref{lem: r->1 power} later on. Recall that $r^T$ is given in \eqref{def: ratio log}.

\begin{lem}\label{lemma: r_T-1 est}
It holds that
 \[
  \lim_{T\rightarrow \infty} \expec^{\prob^T}\bra{\left|1-\fR(X^{1,T}_T)(r^T_T)^{p-1}\right| \left|r^T_T-1\right|}=0.
 \]
\end{lem}

\begin{proof}
 To ease notation, set $\fR_T = \fR(X^{1,T}_T)$ and $r_T = r^T_T$. It follows from
 \eqref{eq: first order cond} that the two inequalities $\expec^\prob[(X^{0,T}_T)^{p-1}(X^{1,T}_T - X^{0,T}_T)]\leq 0$ and $\expec^\prob[U'(X^{1,T}_T)(X^{0,T}_T -
 X^{1,T}_T)]\leq 0$ hold. Summing these two inequalities, it follows that
\begin{equation*}
\begin{split}
0&\geq
\expec^{\prob}\bra{\pare{(X^{0,T}_T)^{p-1}-U'(X^{1,T}_T)}\pare{X^{1,T}_T-X^{0,T}_T}},\\
&=
\expec^{\prob}\bra{(X^{0,T}_T)^{p-1}\pare{1-\frac{U'(X^{1,T}_T)}{(X^{1,T}_T)^{p-1}}\frac{(X^{1,T}_T)^{p-1}}{(X^{0,T}_T)^{p-1}}}\pare{X^{1,T}_T-X^{0,T}_T}},\\
&=\expec^{\prob}\bra{(X^{0,T}_T)^{p}\pare{1-\fR_Tr_T^{p-1}}\pare{r_T - 1}}.
\end{split}
\end{equation*}
After changing to the measure $\prob^T$, the previous inequality reads
 \begin{equation*}
  \expec^{\prob^T} \bra{\pare{1- \fR_T \,r_T^{p-1}}(r_T-1)} \leq 0.
 \end{equation*}
Note that $\pare{1- \fR_T \,r_T^{p-1}}(r_T-1)\leq 0$ if and only if $\fR_T^{1/(1-p)} \leq r_T \leq 1$ or $1\leq r_T \leq \fR_T^{1/(1-p)}$, hence
 \begin{equation}\label{eq: rT-1 est}
 \expec^{\prob^T} \bra{\left|1- \fR_T \,r_T^{p-1}\right||r_T-1|} \leq 2\, \expec^{\prob^T}\bra{\pare{1- \fR_T r_T^{p-1}}(1-r_T) \,\indic_{\left\{R_T^{1/(1-p)} \leq r_T \leq 1 \text{ or } 1\leq r_T \leq R^{1/(1-p)}\right\}}}.
 \end{equation}
 Let us estimate the right-hand-side expectation on $\{\fR_T^{1/(1-p)} \leq r_T \leq 1\}$ and $\{1\leq r_T \leq \fR_T^{1/(1-p)}\}$ separately. On the first set, note that $(1- \fR_T r_T^{p-1})(1-r_T) \leq (1- \fR_T)(1-\fR^{1/(1-p)}_T)$. Then
\begin{align*}
 &\expec^{\prob^T}\bra{(1- \fR_T r_T^{p-1})(1-r_T) \, \indic_{\{\fR_T^{1/(1-p)} \leq r_T\leq 1\}}} \\
  & \leq \expec^{\prob^T}\bra{(1-\fR_T)(1- \fR_T^{1/(1-p)}) \,\indic_{\{\fR_T\leq 1\}}}\\
  &  \leq \prob^T(\fR_T \leq 1-\epsilon) + \expec^{\prob^T}\bra{(1- \fR_T)(1- \fR_T^{1/(1-p)})\,\indic_{\{1-\epsilon \leq \fR_T \leq 1\}}}\\
  & \leq \prob^T(\fR_T \leq 1-\epsilon) + \epsilon (1- (1-\epsilon)^{1/(1-p)}).
\end{align*}
 Sending $T\rightarrow \infty$ then $\epsilon\downarrow 0$ and using Lemma \ref{lemma: asy RX1}, it follows that
 \begin{equation}\label{eq: rT-1 est 1}
 \lim_{T\rightarrow \infty} \expec^{\prob^T}\bra{\pare{1- \fR_T r_T^{p-1}}(1-r_T) \,\indic_{\{R_T^{1/(1-p)} \leq r_T \leq 1\}}}=0.
 \end{equation}

 On the set $\{1\leq r_T \leq \fR_T^{1/(1-p)}\}$, observe that $\fR_T r_T^{p-1} + r_T \geq 2$. Then on the same set,
 \[
  \pare{1- \fR_T r_T^{p-1}}(1-r_T) = \fR_T r_T^p - \fR_T r_T^{p-1} - r_T +1 \leq \fR_T r_T^p-1.
 \]
 Therefore
 \[
 \begin{split}
  &\expec^{\prob^T}\bra{\pare{1- \fR_T r_T^{p-1}}(1-r_T)\, \indic_{\{1\leq r_T \leq \fR_T^{1/(1-p)}\}}}\\
  &\leq \expec^{\prob^T}\bra{\pare{1- \fR_T r_T^{p-1}}(1-r_T)\, \indic_{\{1\leq r_T \leq \fR_T^{1/(1-p)}, \fR_T\leq 1+\epsilon\}}} + \expec^{\prob^T}\bra{\pare{\fR_T r_T^p -1}\, \indic_{\{1\leq r_T \leq \fR_T^{1/(1-p)}, 1+\epsilon <\fR_T\}}}\\
  &  =: J_1 + J_2.
 \end{split}
 \]
 In the previous equation, $J_1 \leq \epsilon ((1+ \epsilon)^{1/(1-p)} -1)$. Let us focus on $J_2$ in what follows. Since
 \[
  J_2 \leq \expec^{\prob^T} \bra{\pare{\fR_T r_T^p-1} \indic_{\{1+ \epsilon < \fR_T\}}} = \expec^{\prob^T}\bra{\fR_T r_T^p \indic_{\{1+ \epsilon < \fR_T\}}} - \prob^T(1+ \epsilon < \fR_T),
 \]
 and $\lim_{T\rightarrow \infty} \prob^T(1+\epsilon < \fR_T) =0$ from Lemma \ref{lemma: asy RX1},
 it suffices to estimate the first term in the previous inequality. To this
 end, note from \eqref{ass: conv} that $\{1+\epsilon < \fR_T\} \subset \{X^{1,T}_T \leq M\}$, for some $M$ depending on $\epsilon$. Then
 \[
 \begin{split}
  \expec^{\prob^T}\bra{\fR_T r_T^p \, \indic_{\{1+ \epsilon < \fR_T\}}} &\leq \expec^{\prob^T} \bra{\fR_T r_T^p \, \indic_{\{X^{1,T}_T\leq M\}}}= \frac{\expec^\prob\bra{U'(X^{1,T}_T) X^{1,T}_T \,\indic_{\{X^{1,T}_T \leq M\}}}}{\expec^\prob[(X^{0,T}_T)^p]} \\
  &= \frac{y^{1,T}}{y^{0,T}}\expec^\prob\bra{Y^{1,T}_T X^{1,T}_T \,\indic_{\{X^{1,T}_T\leq M\}}} .
 \end{split}
 \]
 Introduce a probability measure $\prob^{1,T}$ via
 \[
  \frac{d\prob^{1,T}}{d\prob} = Y^{1,T}_T X^{1,T}_T.
 \]
A line of reasoning similar to that in iii) of Remark \ref{rem: long run measure} shows that $X^{1,T}$ has the num\'{e}raire property under
$\prob^{1,T}$. Thus, the argument in Lemma \ref{lemma: opt unbounded power}
applied to $X^{1,T}$ and $\prob^{1,T}$ implies that $\lim_{T\rightarrow \infty} \prob^{1,T}(X^{1,T} \geq M) =1$. The previous convergence, combined with Lemma \ref{lemma: ratio lagrangian}, then implies
 \[
 \frac{y^{1,T}}{y^{0,T}}\expec^\prob\bra{Y^{1,T}_T X^{1,T}_T \,\indic_{\{X^{1,T}_T\leq M\}}}  = \frac{y^{1,T}}{y^{0,T}}\prob^{1,T}(X^{1,T}_T \leq M) \rightarrow 0, \quad \text{ as } T\rightarrow \infty.
 \]
 Now, combining estimates on $J_1$ and $J_2$, and utilizing Lemma \ref{lemma: asy RX1}, we obtain from sending $T\rightarrow \infty$ then $\epsilon \downarrow 0$ that
 \[
  \lim_{T\rightarrow \infty} \expec^{\prob^T}\bra{\pare{1- \fR_T r_T^{p-1}}(1-r_T)\, \indic_{\{1\leq r_T \leq \fR_T^{1/(1-p)}\}}} =0.
 \]
 Combining the previous convergence with \eqref{eq: rT-1 est 1}, the statement now follows from \eqref{eq: rT-1 est}.
\end{proof}

The previous result implies that $r^T_T \rightarrow 1$ under $\prob^T$.

\begin{lem}\label{lem: r^T->1 prob}
It holds that
 \begin{equation*}
  \lim_{T\to \infty} \prob^T (|r^T_T-1|\geq \epsilon) =0, \quad \text{ for any } \epsilon >0.
 \end{equation*}
\end{lem}

\begin{proof}
As in the proof of Lemma \ref{lemma: r_T-1 est}, set $\fR_T = \fR(X^{1,T}_T)$
and $r_T = r^T_T$. Fix $\epsilon\in (0,1)$ and consider the set
 \[
  D^T = \set{|r_T-1|\geq \epsilon, (1-\epsilon)^{\frac{1-p}{2}} \leq \fR_T \leq (1+\epsilon)^{\frac{1-p}{2}}}.
 \]
 In the following, we estimate the lower bound of $\left|1- \fR_Tr_T^{p-1}\right|$ on $D^T$ for the cases $r_T\geq 1+\epsilon$ and $r_T \leq 1-\epsilon$ separately.

 For $r_T\geq 1+\epsilon$, on $D^T$ we have $\fR_Tr_T^{p-1} \leq (1+\epsilon)^{\frac{p-1}{2}}<1$, whence
 \[
  1- \fR_Tr_T^{p-1} \geq 1-(1+\epsilon)^{\frac{p-1}{2}} >0;
 \]
 For $r_T\leq 1-\epsilon$, on $D^T$ we have $\fR_Tr_T^{p-1} \geq (1-\epsilon)^{\frac{p-1}{2}}>1$ on $D^T$ whence
 \[
  1- \fR_Tr_T^{p-1} \leq 1-(1-\epsilon)^{\frac{p-1}{2}} <0.
 \]
 Denote $\eta = \min\set{1-(1+\epsilon)^{\frac{p-1}{2}}, -1+(1-\epsilon)^{\frac{p-1}{2}}}$. In either of the above cases, $\left|1- \fR_Tr_T^{p-1}\right| \geq \eta$, therefore
 \[
  \expec^{\prob^T} \bra{\left|1- \fR_Tr_T^{p-1}\right| |r_T-1|} \geq \epsilon \, \eta\, \prob^T(D^T).
 \]
Combining the previous inequality with Lemma \ref{lemma: r_T-1 est}, it follows that
 \begin{equation*}
  \lim_{T\to \infty} \prob^T(D^T)=0.
 \end{equation*}
Now, combining the previous convergence with the second statement in Lemma \ref{lemma: asy RX1}, we conclude the proof.
\end{proof}

Now we are ready to prove the main result of this subsection.

\begin{prop}\label{lem: r->1 power}
Under Assumptions \ref{ass: utility} - \ref{ass: wellpose},
 \[
  \lim_{T\rightarrow \infty} \expec^{\prob^T} \bra{|r^T_T-1|} =0.
 \]
\end{prop}

\begin{proof}

As in the previous Lemmas, set $r_T = r^T_T$. The proof consists of the
following two steps, whose combination confirms the claim. Note that for $p=0$,
$\prob^T$ below is exactly $\prob$ and hence convergence takes place under the
physical measure.

 \vspace{2mm}
 \noindent\underline{Step 1:} Establish that
 \begin{equation}\label{eq: pow L1 est 1}
  \lim_{T\rightarrow \infty} \expec^{\prob^T}\bra{\left|r_T-1\right|\, \indic_{\{r_T\leq N\}}} = 0, \quad \text{ for any } N>2.
 \end{equation}
 To this end, for any $\epsilon>0$, we have
 \[
  \begin{split}
   \expec^{\prob^T}\bra{\left|r_T-1\right| \indic_{\{r_T \leq N\}}} &= \expec^{\prob^T}\bra{|r_T-1| \indic_{\set{r_T\leq N, \,|r_T-1|\leq \epsilon}}} + \expec^{\prob^T}\bra{|r_T-1|  \indic_{\set{r_T\leq N, \, |r_T-1|> \epsilon}}} \\
   & \quad \leq \epsilon + (N-1) \, \prob^T\pare{|r_T-1| > \epsilon}.
  \end{split}
 \]
 As $T \rightarrow \infty$, \eqref{eq: pow L1 est 1} follows from Lemma \ref{lem: r^T->1 prob} and the arbitrary choice of $\epsilon$.

 \vspace{2mm}
 \noindent\underline{Step 2:} We also establish that
 \begin{equation}\label{eq: pow L1 est 2}
  \lim_{T\rightarrow \infty} \expec^{\prob^T} \bra{|r_T-1| \, \indic_{\set{r_T>N}}} =0, \quad \text{ for any } N>2.
 \end{equation}
 To this end,
 \[
  \expec^{\prob^T} \bra{|r_T-1| \, \indic_{\set{r_T>N}}} \leq \expec^{\prob^T}\bra{r_T \indic_{\{r_T > N\}}} = \expec^{\prob^T}[r_T] - \expec^{\prob^T}[(r_T-1) \indic_{\{r_T \leq N\}}] - \prob^T(r_T \leq N).
 \]
 Note that $\expec^{\prob^T}[r_T] \leq 1$ due to the num\'eraire property of $X^{0,T}$ under $\prob^T$ (cf. Remark \ref{rem: long run measure} item iii)), $\lim_{T\rightarrow \infty} \expec^{\prob^T}[(r_T-1) \indic_{\{r_T\leq N\}}] = 0$ from Step 1, and $\lim_{T\rightarrow \infty} \prob^T(r_T\leq N) =1$ from Lemma \ref{lem: r^T->1 prob}, therefore
 \[
  0 \leq \limsup_{T\rightarrow \infty} \expec^{\prob^T} \bra{|r_T-1| \, \indic_{\set{r_T>N}}} \leq 1-0-1 = 0,
 \]
 which confirms \eqref{eq: pow L1 est 2}.
\end{proof}

\subsection{Convergence of Wealth Processes}

The following lemma bridges this transition from the convergence of optimal payoffs to the convergence of their wealth processes.
\begin{lem}\label{lemma: stoch-log conv}
 Consider a sequence $(r^T)_{T\in \Real_+}$ of \cadlag \, processes and a sequence $(\prob^T)_{T\in \Real_+}$ of probability measures, such that:
 \begin{enumerate}
  \item[i)] For each $T\in \Real_+$, $r^T$ is defined on $[0,T]$ with $r^T_0
    =1$ and $r^T_t >0$ for all $t\leq T$, $\prob^T$-a.s..
  \item[ii)] Each $r^T$ is a $\prob^T$-supermartingale on $[0,T]$
  \item[iii)] $\lim_{T\to \infty} \expec^{\prob^T}\bra{|r^T_T -1|} =0$.
 \end{enumerate}
 Then:
 \begin{enumerate}
  \item[a)] $\lim_{T\to \infty} \prob^{T}\pare{\sup_{u\in [0,T]} \left|r^T_u - 1\right| \geq \epsilon} =0$, for any $\epsilon>0$.
  \item[b)] Define $L^T := \int_0^{\cdot}(1/ r^T_{t-}) \, dr^T_t$, i.e., $L^T$ is the stochastic logarithm of $r^T$, for each $T\in \Real_+$. Then $\lim_{T\to \infty} \prob^T \pare{\bra{L^T, L^T}_T \geq \epsilon} =0$, for any $\epsilon>0$, where $[\cdot, \cdot]_T$ is the square bracket on $[0,T]$.
 \end{enumerate}
\end{lem}

\begin{proof}
This result follows from Theorem~2.5 and Remark~2.6 in \cite{Kardaras}. (Note that Theorem~2.5 in \cite{Kardaras} is stated under a fixed probability $\overline{\prob}$ and on a fixed time interval $[0,T]$, but its proof remains valid for a sequence of probability measures $(\prob^T)_{T\in \Real_+}$ and a family of time intervals $( [0,T] )_{T \in \Real_+}$.)
\end{proof}

Combining the Lemma \ref{lemma: stoch-log conv} with Proposition \ref{lem: r->1 power}, Theorem \ref{thm: opt-port conv} is proved as follows.

\begin{proof}[Proof of Theorem~\ref{thm: opt-port conv}]
 The statements follow from Lemma~\ref{lemma: stoch-log conv} directly, once we check that the assumptions of Lemma~\ref{lemma: stoch-log conv} are satisfied. First, $r_0^T=1$ since both investors have the same initial capital. Second, assuming $r^T_{\cdot}$ being a $\prob^T$-supermartingale for a moment, then $r_t^T>0$, $\prob^T$-a.s., for any $t\leq T$, because $r^T_T>0$, $\prob^T$-a.s. (see Remark \ref{rem: long run measure} i)). Third, $\lim_{T\to \infty} \expec^{\prob^T} \bra{|r^T_T-1|} =0$ is the result of Proposition \ref{lem: r->1 power}. Hence it remains to show that $r^T_{\cdot}$ is a $\prob^T$-supermartingale. To this end, it suffices to show that:
\begin{equation}\label{eq: exp f/f^0} \expec^{\prob^T} \bra{X_t / X^{0,T}_t | \F_s} \leq X_s / X^{0,T}_s, \quad \text{ for any } s<t\leq T \text{ and } X\in \mathcal{X}^T.
\end{equation}
Since $X^{0,T}_T>0$ $\prob^T$ a.s., Remark \ref{rem: long run measure} i) implies that both denominators in above inequality are nonzero.
To prove \eqref{eq: exp f/f^0}, fix any $A\in \F_s$, and construct the wealth process $\widetilde{X}\in \mathcal{X}^T$ via
\[
\widetilde{X}_u := \left\{\begin{array}{ll} X^{0,T}_u, & u \in [0,s) \\ {X^{0,T}_s}\frac{X_u}{X_s}  \indic_A + X^{0,T}_u \indic_{\Omega \setminus A}, & u \in [s,t)\\
{X^{0,T}_s}\frac{X_t}{X_s} \frac{X^{0,T}_u}{X^{0,T}_t} \indic_A + X^{0,T}_u \indic_{\Omega \setminus A}, & u\in [t, T] \end{array}\right..
 \]
Noting that
 \[
  \frac{\widetilde{X}_T}{X^{0,T}_T} = \frac{X_s^{0,T}}{X_s} \frac{X_t}{X^{0,T}_t} \indic_A + \indic_{\Omega \setminus A},
 \]
 the claim follows from $\expec^{\prob^T}\bra{\widetilde{X}_T / X^{0,T}_T}\leq 1$ (cf. Remark \ref{rem: long run measure} iii)) and the arbitrary choice of $A$.
\end{proof}

\begin{proof}[Proof of Corollary~\ref{thm: power mayopic}]
First, we prove that $(d\prob^T/d\prob\big|_{\F_t})_{T\geq 0}$ is a constant sequence. For any $t\leq T\leq S$,
 \[
 \begin{split}
  \left. \frac{d\prob^S}{d\prob}\right|_{\F_t} &= \frac{\expec_t^{\prob}\bra{\pare{X_S^{0,S}}^p}}{\expec^{\prob}\bra{\pare{X_S^{0,S}}^p}} = \frac{\expec_t^{\prob}\bra{\pare{X_T^{0,S}}^p \pare{X^{0,S}_S / X^{0,S}_T}^p}}{\expec^{\prob}\bra{\pare{X_T^{0,S}}^p \pare{X^{0,S}_S / X^{0,S}_T}^p}} = \frac{\expec_t^{\prob}\bra{\pare{X_T^{0,S}}^p} \expec^{\prob}_t\bra{\pare{X^{0,S}_S / X^{0,S}_T}^p}}{\expec^{\prob}\bra{\pare{X_T^{0,S}}^p} \expec^{\prob}\bra{\pare{X^{0,S}_S / X^{0,S}_T}^p}}\\
  & = \frac{\expec_t^{\prob}\bra{\pare{X^{0,S}_T}^p}}{{\expec^{\prob}\bra{\pare{X^{0,S}_T}^p}}} = \frac{\expec_t^{\prob}\bra{\pare{X^{0,T}_T}^p}}{{\expec^{\prob}\bra{\pare{X^{0,T}_T}^p}}} = \left. \frac{d\prob^T}{d\prob} \right|_{\F_t}.
 \end{split}
 \]
Here, the third equality follows from the assumption that $X^{0,S}_T$ and $X^{0,S}_S/X^{0,S}_T$ are independent; the fourth equality holds since $X^{0,S}_S/X^{0,S}_T$ is independent of $\F_t$; and the fifth equality holds by the myopic optimality $X^{0,T}_T= X^{0,S}_T$.

Second, note that $d\prob^T/d\prob\big|_{\F_\cdot}$ is a strictly positive martingale; see the discussion after \eqref{def: p^T}, it then induces a probability measure $\tilde{\prob}$, which is equivalent to $\prob$ on $\F_t$. As a result, the rest statements follows from Theorem~\ref{thm: opt-port conv} and Remark~\ref{rem: finite time}  iii) directly.
\end{proof}


\section{Proof of the Turnpike for Diffusions}

This section contains the proofs of the statements in section
\ref{subsec: turnpike for diffusion}, and is broken into four subsections.
The first one details properties of the measures $(\hat{\prob}^{\xi})_{\xi\in\Real^{d}\times E}$. The second
subsection constructs the reduced value function $v^T$ in
\eqref{eq: HJB}. The third section establishes the precise
relations between the conditional densities $d\prob^{T,y}/d\prob^y
|_{\F_t}$, $d\hat{\prob}^{y}/d\prob^y|_{\F_t}$ and the wealth processes $X^{0,T}/\hat{X}$. The last subsection contains the proofs of the main results.
\begin{rem}\label{R:stoch_exp}
To ease notation, denote in the sequel:
\[
  \Exp\pare{\int H \,dW}_{t,s} := \exp\pare{\int_t^s H_u \, dW_u - \frac12
    \int_t^s \|H_u\|^2 \,du} \quad \text{ for any integrand } H \text{ and }
  t\leq s,
\]
and by $\Exp\pare{\int H \,dW}_t = \Exp\pare{\int H \,dW}_{0,t}$.
\end{rem}


\subsection{The long run measure
  $\hat{\prob}^{\xi}$}\label{SS: phat Y prop}

Recall the following terminology from ergodic theory for diffusions (see
\cite{Pinsky} for a more thorough treatment). Let $L$ be as in \eqref{eq: gen op def}. Suppose the martingale problem
for $L$ is well posed on $D$, and denote its solution
by $(\prob^x)_{x\in D}$, with coordinate process $\Xi$. Denote by $L^*$ the
formal adjoint to $L$. Note that under the given regularity conditions $L^*$
is a second order differential operator as well.

$\Xi$ is \emph{recurrent} under $(\prob^x)_{x\in D}$ if $\prob^x(\tau(\epsilon,
y)<\infty) =1$ for any $(x,y)\in D^2$ and $\epsilon>0$, where
$\tau(\epsilon, y) =\inf\set{t\geq 0 \such |\Xi_t - y| \leq \epsilon}$. If
$\Xi$ is recurrent then it is \emph{positive recurrent}, or \emph{ergodic}
if there exists a strictly positive $\varphi^*\in
C^{2,\gamma}(D,\Real^k)$ such that $L^*\varphi^* = 0$ and $\int_{D}
\varphi^*(y)\, dy<\infty$. If $\Xi$ is recurrent, but not positive recurrent,
it is \emph{null recurrent}.


\begin{lem}\label{lemma: hat-P posrec}
Let Assumptions \ref{ass: regular coeffs} and \ref{ass: ptphat} hold. Let
$\cL$ be as in \eqref{eq: cL B def} and define $\cL^{\hat{v},0}$ by
\begin{equation}\label{eq: cl hatv def}
\cL^{\hat{v},0} = \cL + A\frac{\hat{v}_y}{\hat{v}}\partial_y.
\end{equation}
Then, there exists a unique solution $(\hat{\prob}^y_Y)_{y\in E}$ to the martingale problem for
$\cL^{\hat{v},0}$  on $E$. Furthermore, the coordinate mapping process $Y$ is
positive recurrent under $(\hat{\prob}^y_Y)_{y\in E}$ with invariant density
$\mu(y) = \hat{v}^2(y)\hat{m}(y)$, where $\hat{m}$ is defined in \eqref{eq: mnu def}.  Therefore, for all functions $f\in L^1(E,\mu)$ and all $t > 0$:
\begin{equation}\label{eq: l1 ergodic thm}
\lim_{T\uparrow\infty} \expec^{\hat{\prob}^y_Y}\bra{f(Y_{T-t})} = \int_E
f(y)\hat{v}^2(y)\hat{m}(y)dy.
\end{equation}
\end{lem}

\begin{proof}
Since \eqref{ass: recurrent} and \eqref{ass: posrec l1} hols,
\citet*[Theorem 5.1.10, Corollary 5.1.11]{Pinsky} implies that: a) $(\hat{\prob}^y_Y)_{y\in E}$ exists and is unique, b) $Y$ is positive
recurrent under $(\hat{\prob}^y_Y)_{y\in E}$, and c) $Y$ has invariant density
$\hat{v}^2(y)\hat{m}(y)$. That \eqref{eq: l1 ergodic thm} holds for $f\in L^1(E,\mu)$ follows from \citet*[Theorem 1.2 (iii), Eqns (3.29) and (3.30)]{MR1152459}.

\end{proof}


\begin{lem}\label{lemma: WP hat-P}
Let Assumptions \ref{ass: regular coeffs} and \ref{ass: ptphat} hold. Then:
\begin{enumerate}[i)]
\item There exists a unique solution $(\hat{\prob}^{\xi})_{\xi\in\Real^d\times
    E}$ to the martingale problem for $\hat{\cL}$ on $\Real^d\times E$ where
  $\hat{\cL}$ is given in \eqref{eq: phat L def}.
\item $\hat{\prob}^\xi \sim \prob^\xi$ for any $\xi\in \Real^d \times E$.
\end{enumerate}
\end{lem}

\begin{proof}
For any integer $n$ denote by $\Omega^n$ be the space of continuous maps
$\omega: \Real_+ \to \Real^{n}$ and $\B^n$ be the sigma algebra generated by
the coordinate process $\Xi$ defined by $\Xi_t(\omega) = \omega_t$ for
$\omega\in\Omega^n$. By Lemma \ref{lemma: hat-P posrec}, there is a unique
solution $(\hat{\prob}^y_Y)_{y\in E}\in M_1(\Omega^1,\B^1)$  to the martingale
problem on $E$ for the operator $\cL^{\hat{v},0}$ given in \eqref{eq: cl hatv def}. Set $\Omega = \Omega^{d+1},\F = \B^{d+1}$ and $\F_t = \B^{d+1}_{t+},
t\geq 0$.  Let $W^{d}$ denote d-dimensional Wiener measure on the first $d$ coordinates
(along with the associated sigma algebra) and set $B = (\Xi^1,...,\Xi^d)$, $Y
= \Xi^{d+1}$.  For any $z\in\Real^d$ define the processes $X,W$ by
\begin{equation*}
X_t := z - q\int_0^t\bar{\rho}'\sigma'\Sigma^{-1}\left(\mu +
  \delta\Upsilon\frac{\hat{v}_y}{\hat{v}}\right)(Y_s)ds + B_t;\qquad W_t =
  \int_0^ta^{-1}(Y_s)\left(dY_s - b(Y_s)ds\right).
\end{equation*}
Clearly for $\xi = (z,y)$ it follows that,
$\pare{(X,Y),(B,W),(\Omega,\F,(\F_t)_{t\geq 0},W^d\times\hat{\prob}^y_Y)}$
is a weak solution to the SDE associated to the operator $\hat{\cL}$.

Since weak solutions induce solutions to the
martingale problem via Ito's formula, it follows that if $\hat{\prob}^{\xi}\in
M_1(\Omega,\F)$ is defined by $\hat{\prob}^{\xi}(A) = W^d\times\hat{\prob}^y((X,Y)\in A)$ with $A\in \F$, then
$(\hat{\prob}^{\xi})_{\xi\in\Real^d\times E}$ solves the
martingale problem for $\hat{\cL}$.

Part (ii) follows from \cite[Remark
2.6]{Cheridito-Filipovic-Yor} and Lemma \ref{lem: mart prob}. Note that the assumption in
\cite{Cheridito-Filipovic-Yor} is satisfied in view of Assumption
\ref{ass: regular coeffs}, $\hat{v} >0$ and $\hat{v}\in C^2(E)$ in Assumption
\ref{ass: ptphat}.
\end{proof}

\begin{rem}\label{rem: hat-P comp} As in the proof of Lemma \ref{lemma: WP hat-P}, for all $\xi  = (z,y)$ with $z\in\Real^d$ and $y\in E$, if $Y$ denotes the $(d+1)^{st}$ coordinate, then
\begin{equation*}
\hat{\prob}^{\xi}(Y\in A) = \hat{\prob}^{y}_Y(Y\in A);\qquad A\in\B^1.
\end{equation*}
Thus, since $Y$ is positive recurrent under $(\hat{\prob}^y_Y)_{y\in E}$ by Lemma
\ref{lemma: hat-P posrec}, $Y$ is positive recurrent under
$(\hat{\prob}^{\xi})_{\xi\in\Real^d\times E}$ with the same invariant density
$\mu$ as in Lemma \ref{lemma: hat-P posrec}. Therefore, the ergodic result in \eqref{eq: l1 ergodic thm} applies to $\hat{\prob}^{\xi}$ for any $\xi\in\Real^{d}\times E$.
\end{rem}


\subsection{Construction of $v^T$}\label{subsec: v^T}

The solution $v^T(t,y)$ to \eqref{eq: HJB} is constructed from the long-run
  solution $\hat{v}(y)$ of Assumption \ref{ass: ptphat}. Recall that
 $\hat{\prob}^\xi$ is denoted by $\hat{\prob}^y$ for $\xi=(0,y)$.
Now, consider the function $h^T$ defined as:
\begin{equation}\label{def: h^T}
 h^T(t,y) := \expec^{\hat{\prob}^y}\bra{\frac{1}{\hat{v}(Y_{T-t})}}, \qquad (t,y)\in [0,T]\times E.
\end{equation}
The candidate reduced value function is
\begin{equation}\label{def : v^T}
v^T(t,y) := e^{\lambda_c (T-t)} \hat{v}(y) h^T(t,y).
\end{equation}
Thus, $h^T$ is the ratio between $v^T$ and its long-run analogue
$e^{\lambda_c(T- \cdot)}\hat{v}$. The verification result
Proposition~\ref{lemma: HJB verification} below confirms that $v^T$ is a
strictly positive classical solution to \eqref{eq: HJB} and the relation
$u^T(t, x, y) = (x^p/p) (v^T(t,y))^\delta$ holds for $(t,x,y)\in
[0,T]\times \Real_+ \times E$.

As a first step to proving Proposition \ref{lemma: HJB verification}, the
next result characterizes the function $h^T$.  Clearly,
$h^T(t,y)>0$ for $(t,y) \in [0,T]\times E$.


\begin{prop}\label{prop: v^T classical}
Let Assumptions \ref{ass: regular coeffs}, \ref{ass: WP P}, \ref{ass: correl} and \ref{ass: ptphat} hold.
Then $h^T(t,y)<\infty$ for all $(t,y)\in [0,T]\times E$, $h^T(t,y)\in
C^{1,2}((0,T)\times E)$, and $h^T$ satisfies
\begin{equation}\label{eq: pde-h}
\begin{split}
& \partial_t h^T + \cL^{\hat{v},0}h^T =0, \quad (t,y)\in (0,T)\times E,\\
& h^T(T,y) = \frac{1}{\hat{v}(y)}, \quad y\in E,
\end{split}
\end{equation}
where $\cL^{\hat{v},0}$ is defined in \eqref{eq: cl hatv def}. Furthermore, the
process
\begin{equation}\label{eq: h^T mart}
\frac{h^T(t,Y_t)}{h^T(0,y)};\qquad 0\leq t\leq T,
\end{equation}
is a $\hat{\prob}^y$ martingale on $[0,T]$ with constant expectation $1$. Lastly, for all $t > 0$ and
$y\in E$ it follows $\hat{\prob}^y$ almost surely that
\begin{equation}\label{eq: h^T limit val}
\lim_{T\rightarrow\infty}\frac{h^T(t,Y_t)}{h^T(0,y)} = 1.
\end{equation}
\end{prop}

\begin{proof}
The proof consists of several steps.

\vspace{2mm}
\noindent{\underline{Step 1: $h^T(t, y)<\infty$ for all $(t, y) \in [0,T]\times E$}.}
Note that
\eqref{eq: eign eqn} implies
\begin{equation}\label{eq: eqn-1/v}
 \cL^{\hat{v}, 0} \pare{\frac{1}{\hat{v}}} = -\frac{\cL\hat{v}}{\hat{v}^2} = \frac{c-\lambda_c}{\hat{v}}.
\end{equation}
Combining \eqref{eq: eqn-1/v} with $\sup_{y \in E} c(y)<\infty$ in
Assumption \ref{ass: correl}, there exists some $K>0$ such that
\begin{equation*}
\left(\partial_t+\cL^{\hat{v},0}\right)\left(\frac{e^{-K t}}{\hat{v}(y)}\right) \leq 0.
\end{equation*}
Thus, using the strict positivity of $\hat{v}$, Assumption \ref{ass: regular coeffs} and Fatou's lemma, it follows that:
\begin{equation}\label{eq: h^T ev finite}
h^T(t,y) = \expec^{\hat{\prob}^y}\bra{\pare{\hat{v}(Y_{T-t})}^{-1}} \leq
\frac{e^{K(T-t)}}{\hat{v}(y)} < \infty.
\end{equation}

\vspace{2mm}
\noindent \underline{Step 2: $h^T\in C^{1,2}((0,T)\times E)$ satisfies \eqref{eq: pde-h}.}
To this end, the classical version of the Feynman-Kac formula (see Theorem~5.3 in \cite{friedman-stochastic} pp. 148) does not apply directly because a) the operator $\cL^{\hat{v}, 0}$ is not assumed to be uniformly elliptic on $E$, and b)
$(\hat{v})^{-1}$ may grow faster than polynomial near the boundary of $E$. Rather, the statement follows from Theorem~1 in \cite{Heath-Schweizer}, which yields that $h^T$ is a classical
solution of \eqref{eq: pde-h}.

To check that the assumptions of Theorem~1 in \cite{Heath-Schweizer} are satisfied, note that, since
$A$ is locally Lipschitz on $E$ due to Assumption~\ref{ass: regular coeffs},
Lemma 1.1 in \cite{friedman-stochastic} pp. 128 implies that $a$ is also
locally Lipschitz on $E$. On the other hand, the local Lipschitz continuity of
$B+A\hat{v}_y/\hat{v}$ is ensured by Assumption~\ref{ass: regular coeffs} and $\hat{v}\in C^2(E)$. Hence (A1) in \cite{Heath-Schweizer} is
satisfied. Second, (A2) in \cite{Heath-Schweizer} holds thanks to the
well-posedness of the martingale problem for $\hat{\cL}$ on $\Real^d\times E$, as proved in Lemma
\ref{lemma: WP hat-P}. Third, (A3') in
\cite{Heath-Schweizer}, (A3a')-(A3d') are clearly satisfied under our
assumptions.

In order to check (A3e'), it suffices to show that $h^T$ is
continuous in any compact sub-domain of $(0,T)\times E$. To this end, recall
that the domain is $E = (\alpha,\beta)$ for $-\infty \leq\alpha <
\beta\leq\infty$.  Let $\{\alpha_m\}$ and $\{\beta_m\}$ two sequences such that
$\alpha_m < \beta_m$ for all $m$,  $\alpha_m$ strictly decreases to $\alpha$,
and $\beta_m$ strictly increases to $\beta$. Set $E_m = (\alpha_m,\beta_m)$. For each $m$ there exists a function $\psi_m(y)\in
C^{\infty}(E)$ such that a) $\psi_m(y)\leq 1$, b) $\psi_m(y) = 1$ on $E_m$, and c)
$\psi_m(y) = 0$ on $E\cap E^c_{m+1}$.  To construct such $\psi_m$ let $\eps_m =
\frac{1}{3}\min\left\{\beta_{m+1}-\beta_m, \alpha_m - \alpha_{m+1}\right\}$
and then take
\begin{equation*}
\psi_m(y) = \eta_{\eps_m} * \indic_{\{\alpha_m - \eps_m, \beta_m + \eps_m\}}(y),
\end{equation*}
where $\eta_{\eps_m}$ is the standard mollifier and $*$ is the convolution operator.  Define the functions $f_m$
and $h^{T,m}$ by
\begin{equation*}
f_m(y) = \frac{\psi_m(y)}{\hat{v}(y)}\quad \text{ and } \quad h^{T,m}(t,y) = \expec^{\hat{P}^y}\bra{f_m(Y_{T-t})}.
\end{equation*}
By construction, for all $y\in E$, $\uparrow\lim_{m\uparrow\infty} f_m(y) =
(\hat{v}(y))^{-1}$.  It then follows from the monotone convergence theorem and \eqref{eq: h^T ev finite} that $\lim_{m\uparrow\infty} h^{T,m}(t,y) = h^{T}(t,y)$.

Since $\hat{v}\in C^2(E)$ and $\hat{v}>0$, each $f_m(y)\in C^{2}(E)$ is bounded. It then follows from the Feller property
for $\hat{\prob}^y$ (see Theorem 1.13.1 in \cite{Pinsky}) that $h^{T,m}$ is
continuous in $y$.  On the other hand, by construction of $f_m$ and \eqref{eq: eqn-1/v}, there exists a constant $K_m>0$ such that
\begin{equation}\label{eq: der est}
a|\dot{f}_m| \leq K_m,\qquad \left|\cL^{\hat{v},0} f_m\right| \leq K_m, \quad \text{ on } E.
\end{equation}
Moreover, Ito's formula gives that, for any $0\leq s\leq t\leq T$,
\[
 f_m(Y_t) = f_m(Y_s) + \int_s^t \cL^{\hat{v}, 0} f_m(Y_u) \, du + \int_s^t a \dot f_n(Y_u) \,d\hat{W}_u.
\]
Combining the previous equation with estimates in \eqref{eq: der est}, it follows that:
\[
 \sup_{y\in E} \left|\expec^{\hat{\prob}^y} \bra{f_m(Y_t) - f_m(Y_s)} \right| \leq K_m(t-s).
\]
Therefore, $h^{T,m}$ is uniformly continuous in $t$. Combining with the continuity of $h^{T,m}$ in $y$, we conclude that $h^{T,m}$ is jointly continuous in $(t,y)$ on $[0,T]\times E$.

Note that the operator $\cL^{\hat{v},0}$ is uniformly elliptic in the
parabolic domain $(0,T)\times E_m $. It then follows from a straightforward calculation that $h^{T,m}$ satisfies the differential equation:
\begin{equation*}
\begin{split}
\partial_t h^{T,m} + \cL^{\hat{v}, 0} h^{T,m} &=0\quad (t,y)\in (0,T) \times
E_m.\\
\end{split}
\end{equation*}
Note that $(h^{T,m})_{m \geq 0}$ is uniformly bounded from above by $h^T$, which is finite on $[0,T]\times E_m$.
Appealing to the
 \emph{interior Schauder estimate} (see e.g. Theorem~15 in
 \cite{friedman-parabolic} pp. 80), there exists a subsequence
 $(h^{T,m'})_{m'}$ which converges to $h^T$ uniformly in $(0,T)\times D$ for
 any compact sub-domain $D$ of $E_m$. Since each $h^{T, m'}$ is continuous and the
 convergence is uniform, we confirm that $h^T$ is continuous in  $(0,T)\times
 D$. Since the choice of $D$ is arbitrary in $E_m$,  (A3e') in
 \cite{Heath-Schweizer} is satisfied.  This proves that $h^T\in
 C^{1,2}((0,T)\times E)$ and satisfies \eqref{eq: pde-h}.

\vspace{2mm}
\noindent\underline{Step 3: Remaining statements.}
By definition of the martingale problem, the process in \eqref{eq: h^T mart}
is a non-negative local martingale, and hence a super-martingale. Furthermore, for $y\in E$, by
construction of $h^T$
\begin{equation*}
\expec^{\hat{\prob}^y}\bra{\frac{h^T(T,Y_T)}{h^T(0,y)}} =
\frac{1}{\expec^{\hat{\prob}^y}\bra{\hat{v}(Y_T)^{-1}}}\expec^{\hat{\prob}^y}\bra{\hat{v}(Y_T)^{-1}}
  = 1,
\end{equation*}
proving the martingale property on $[0,T]$.  Lastly, \eqref{eq: h^T limit val}
follows from \eqref{eq: l1 ergodic thm} in Lemma \ref{lemma: hat-P posrec} and
Remark \ref{rem: hat-P comp}, since \eqref{ass: posrec l1} in Assumption
\ref{ass: ptphat} implies that $(\hat{v})^{-1}\in L^1(E,\hat{v}^2\hat{m})$. Thus, for
all $t > 0, y\in E$,
\begin{equation*}
\lim_{T\rightarrow\infty} h^T(t,y) = \int_E\hat{v}(z)\hat{m}(z)dz.
\end{equation*}
which gives the result by taking $y = Y_t$ for a fixed $t$.
\end{proof}


The next step towards the verification result in Proposition
\ref{lemma: HJB verification} is to connect solutions $v^T$ to the PDE
in \eqref{eq: HJB} to the value function $u^T$ of \eqref{def: u}. This task requires additional notation.

Let Assumption \ref{ass: regular coeffs} and \ref{ass: WP P} hold. For any strictly positive $w(t,y)\in C^{1,2}((0,T)\times E)$, define the process
\begin{equation}\label{eq: D_hatpw_p}
D^w_t := \Exp\left(\int\left(-q\Upsilon'\Sigma^{-1}\mu + A\frac{w_y}{w}\right)'\frac{1}{a}dW -
  q\int\left(\Sigma^{-1}\mu + \Sigma^{-1}\Upsilon\delta\frac{w_y}{w}\right)'\sigma\bar{\rho}dB\right)_t,
\end{equation}
Here, $(B,W)$ is a $(d+1)$-dimensional $\prob^y$ Brownian Motion.  $B$ is formed
by the first $d$ coordinates of $\Xi$ and $W$ is from Remark \ref{rem: mart prob}.  Note that $B$ is $\prob^y$ independent of $Y$.  For $t\leq s\leq T$ define
$D^w_{t,s}$ as in Remark \ref{R:stoch_exp}. Denote the portfolio
\begin{equation}\label{eq: pi_w}
\pi^w=\frac{1}{1-p}\Sigma^{-1}\left(\mu + \delta\Upsilon
\frac{w_y}{w}\right),
\end{equation}
evaluated at $(t,Y_t)$, and by $X^{\pi,w}$ the corresponding wealth process. Set
$\eta^w = \delta w_y/w$ (also evaluated at $(t,Y_t)$), and define the process $M^{\eta,w}$ via
\begin{equation}\label{eq: eta_w}
M^{\eta,w}_t = e^{-\int_0^trd\tau}\Exp\left(\int\left(-\Upsilon'\Sigma^{-1}\mu +
    (A-\Upsilon'\Sigma^{-1}\Upsilon)\eta^w\right)'\frac{1}{a}dW -
  \int\left(\Sigma^{-1}\mu + \Sigma^{-1}\Upsilon\eta^w\right)'\sigma\bar{\rho}dB\right)_t.
\end{equation}
The following proposition is crucial to check the optimality of both
finite horizon and long-run optimal portfolios, and to compare their terminal wealths.  A similar statement for the long-run limit is in \cite[Theorem 7]{Guasoni-Robertson}.


\begin{prop}\label{prop: Ito-Girsanov result}
Let Assumptions \ref{ass: regular coeffs} and \ref{ass: WP P} hold. Assume
there exists a function $w: [0,T]\times E \to \Real$ and a constant
$\lambda\in\Real$, such that $w\in C^{1,2}((0,T)\times E,\Real)$ is strictly positive and
satisfies the differential expression
\begin{equation*} 
\partial_t w + \cL w + (c-\lambda)w = 0, \quad (t,y) \in (0,T) \times E.
\end{equation*}
Then the following conclusions hold:
\begin{enumerate}[(i)]
\item For all admissible portfolios $\pi$ and all $y\in E$, the process
  $X^{\pi}M^{\eta,w}$ is a non-negative supermartingale under $\prob^y$ where $M^{\eta,w}$ is given in \eqref{eq: eta_w}.
\item For all $t\leq s\leq T$ and $y\in E$ the processes $X^{\pi,w}$ and
  $M^{\eta,w}$ satisfy the $\prob^y$ almost sure identities
\begin{equation}\label{eq: primal dual as eq}
\begin{split}
\left(X^{\pi,w}_s\right)^p &= \left(X^{\pi,w}_t\right)^{p}\left(w(t,Y_t)e^{\lambda(s-t)}\right)^{\delta}D^w_{t,s}w(s,Y_s)^{-\delta},\\
\left(M^{\eta,w}_s\right)^q &= \left(M^{\eta,w}_t\right)^q\left(w(t,Y_t)e^{\lambda(s-t)}\right)^{\frac{\delta}{1-p}}D^w_{t,s}w(s,Y_s)^{-\frac{\delta}{1-p}},\\
\end{split}
\end{equation}
where $\pi^w$, $M^{\eta,w}$ and $D^{w}$ are as in \eqref{eq: pi_w}, \eqref{eq: eta_w} and \eqref{eq: D_hatpw_p} respectively.
\end{enumerate}
\end{prop}

\begin{proof}
Given $w$, it is clear, using stochastic
integration by parts, that for $i=1,\dots, d$ the process
$M^{\eta,w} S^{i}$ is a non-negative supermartingale under $\prob^y$ for any
$y\in E$. Thus, part $(i)$ follows.

It remains the show the almost-sure identities. To this end fix $t\leq s\leq
T$.  By \eqref{eq: wealth dynamics} it follows that
\begin{equation}\label{eq: primal dual rel}
\begin{split}
\frac{\left(X^{\pi,w}_s\right)^p}{
\left(X^{\pi,w}_t\right)^p}&=\exp\left(\int_t^s\left(p\mu'\pi^w + pr
    -\frac{p}{2}(\pi^w)'\Sigma\pi^w\right)d\tau +
  p\int_t^s(\pi^w)'\sigma dZ_\tau\right),\\
\frac{\left(M^{\eta,w}_s\right)^q}{\left(M^{\eta,w}_t\right)^q}&=e^{-q\int_t^srd\tau}\Exp\left(\int\left(-\Upsilon'\Sigma^{-1}\mu +
    (A-\Upsilon'\Sigma^{-1}\Upsilon)\eta^w\right)'\frac{1}{a}dW -
  \int\left(\Sigma^{-1}\mu +
    \Sigma^{-1}\Upsilon\eta^w\right)'\sigma\bar{\rho}dB\right)^q_{t,s}.
\end{split}
\end{equation}
Taking logarithms and expanding $Z = \rho W + \bar{\rho} B$, the first equality is equivalent to
\begin{multline}\label{eq: trading_equal}
\int_t^s\left(p\mu'\pi^w + pr
    -\frac{p}{2}(\pi^w)'\Sigma\pi^w\right)d\tau +
  p\int_t^s(\pi^w)'\sigma\rho dW_\tau +
  p\int_t^s(\pi^w)'\sigma\bar{\rho} dB_\tau\\
= \delta\log w(t,Y_t) + \delta\lambda(s-t)
  + \log D^w_{t,s} - \delta\log w(s,Y_s).
\end{multline}
Multidimensional notation makes calculations more transparent. Set $\omega:= \delta\log w$.  Then $\omega$ solves the quasi-linear differential expression
\begin{equation*}
\partial_t \omega + \cL \omega + \frac{1}{2}\nabla
\omega'\left(A-q\Upsilon'\Sigma^{-1}\Upsilon\right)\nabla \omega + \delta(c-\lambda) = 0, \quad (t,y) \in (0,T) \times E,
\end{equation*}
because $\delta^{-1} A =
A-q\Upsilon'\Sigma^{-1}\Upsilon$. Since $\cL = (1/2)A\partial_{yy} +
(b-q\Upsilon'\Sigma^{-1}\mu)\partial_y$, Ito's formula implies that
\emph{under $\prob^y$}, after expanding $c=(1/\delta)\left(pr - (1/2)q\mu'\Sigma^{-1}\mu\right)$,
\begin{multline}\label{E:temp_p1}
\delta\log w(s,Y_s) - \delta\log w(t,Y_t) =\\
\int_t^s\left(q\mu'\Sigma^{-1}\Upsilon\nabla \omega -\frac{1}{2}\nabla
\omega'(A-q\Upsilon'\Sigma^{-1}\Upsilon)\nabla \omega + \delta\lambda - \left(pr-\frac{1}{2}q\mu'\Sigma^{-1}\mu\right)\right)d\tau + \int_t^s\nabla \omega'a dW_\tau.
\end{multline}
Again, using $\delta^{-1} A =
A-q\Upsilon'\Sigma^{-1}\Upsilon$ and noting that $\delta w_y/w = \omega_y$, after some simplifications it follows that:
\begin{equation}\label{E:temp_p2}
\begin{split}
\log D^w_{t,s} &=  \int_t^s\left(-q\mu'\Sigma^{-1}\Upsilon +
    \nabla\omega'(A-q\Upsilon'\Sigma^{-1}\Upsilon)\right)\frac{1}{a}dW_\tau\\
 &- q\int_t^s\left(\mu'\Sigma^{-1} + \nabla\omega\Upsilon'\Sigma^{-1}\right)\sigma\bar{\rho}dB_\tau\\
&+\int_t^s\left(-\frac{1}{2}q^2\mu'\Sigma^{-1}\mu +
  q(1-q)\mu'\Sigma^{-1}\Upsilon\nabla\omega-\frac{1}{2}\nabla\omega'\left(A-(2q-q^2)\Upsilon'\Sigma^{-1}\Upsilon\right)\nabla\omega\right)d\tau.
\end{split}
\end{equation}
Lastly, plugging in for $\pi^w$ yields
\begin{equation}\label{E:temp_p3}
\begin{split}
&\int_t^s\left(p\mu'\pi^w + pr
    -\frac{p}{2}(\pi^w)'\Sigma\pi^w\right)d\tau +
  p\int_t^s(\pi^w)'\sigma\rho dW_\tau +
  p\int_t^s(\pi^w)'\sigma\bar{\rho} dB_\tau\\
&=\int_t^s\left(pr -\frac{1}{2}q(1+q)\mu'\Sigma^{-1}\mu -
    q^2\mu'\Sigma^{-1}\Upsilon\nabla \omega - \frac{1}{2}q(q-1)\nabla
    \omega'\Upsilon'\Sigma^{-1}\Upsilon\nabla \omega\right)d\tau\\
&\qquad - q\int_t^s\left(\mu' +
      \nabla\omega'\Upsilon'\right)\Sigma^{-1}\sigma\rho dW_{\tau} -q\int_t^s\left(\mu' +
      \nabla \omega'\Upsilon'\right)\Sigma^{-1}\sigma\bar{\rho}dB_{\tau}.
\end{split}
\end{equation}
Now, using \eqref{E:temp_p1}, \eqref{E:temp_p2} and \eqref{E:temp_p3}, the
equality in \eqref{eq: trading_equal} follows by matching the respective
$d\tau$, $dW$ and $dB$ terms.

The proof for the second identity in \eqref{eq: primal dual as eq} is similar.  Given \eqref{eq: primal dual rel}, it suffices to show that, by
taking logarithms
\begin{multline}\label{eq: dual equal}
-q\int_t^s rd\tau + q\log \Exp\left(\int\left(-\Upsilon'\Sigma^{-1}\mu +
    (A-\Upsilon'\Sigma^{-1}\Upsilon)\eta^w\right)\frac{1}{a}dW -
  \int\left(\Sigma^{-1}\mu +
    \Sigma^{-1}\Upsilon\eta^w\right)\sigma\bar{\rho}dB\right)_{t,s} \\
= \frac{\delta}{1-p}\log w(t,Y_t) + \frac{\delta}{1-p}\lambda(s-t)
  + \log D^w_{t,s} - \frac{\delta}{1-p}\log w(s,Y_s).
\end{multline}
The equality in \eqref{E:temp_p1} (multiplied by $1/(1-p)$), combined
with that in \eqref{E:temp_p2} yield an expression for the right hand side of the
above equation in terms of integrals from $s$ to $t$ of $d\tau,dW$ and $dB$.  As
for the left hand side, a lengthy calculation shows that
\begin{equation}\label{E:temp_p4}
\begin{split}
&-q\int_t^s rd\tau+q\log \Exp\left(\int\left(-\Upsilon'\Sigma^{-1}\mu +
    (A-\Upsilon'\Sigma^{-1}\Upsilon)\eta^w\right)'\frac{1}{a}dW -
  \int\left(\Sigma^{-1}\mu +
    \Sigma^{-1}\Upsilon\eta^w\right)'\sigma\bar{\rho}dB\right)_{t,s} \\
=& \int_t^s\left(-qr - \frac{1}{2}q\mu'\Sigma^{-1}\mu -
  \frac{1}{2}q\nabla\omega'(A-\Upsilon'\Sigma^{-1}\Upsilon)\nabla\omega\right)d\tau\\
&+q\int_t^s\left(-\mu'\Sigma^{-1}\Upsilon +
  \nabla\omega'(A-\Upsilon'\Sigma^{-1}\Upsilon)\right)\frac{1}{a}dW_\tau\\
&-q\int_t^2\left(\mu'\Sigma^{-1} +
  \nabla\omega'\Upsilon'\Sigma^{-1}\right)\sigma\bar{\rho}dB_{\tau}.
\end{split}
\end{equation}
Thus, using \eqref{E:temp_p1}, \eqref{E:temp_p2} and \eqref{E:temp_p4}, the
equality in \eqref{eq: dual equal} follows by matching $d\tau,dW$ and $dB$ terms.
\end{proof}

The verification result for the finite horizon problem now follows.

\begin{prop}\label{lemma: HJB verification}
Let Assumptions \ref{ass: regular coeffs}, \ref{ass: WP P}, \ref{ass: correl},
and \ref{ass: ptphat} hold. Define $v^T$ by \eqref{def : v^T}.
Then:
\begin{enumerate}[(i)]
 \item  $v^T>0$, $v^T \in C^{1,2}((0,T)\times E)$, and it solves \eqref{eq: HJB}.
 \item  $u^T(t,x,y) = \frac{x^p}{p} \pare{v^T(t,y)}^\delta$ on $[0,T]\times
   \Real_+ \times E$ and $\pi^T$ in \eqref{eq: finite opt} is the optimal portfolio.
\end{enumerate}
\end{prop}

\begin{proof}
Clearly, the positivity of $h^T$ and $\hat{v}$ yield that of
$v^T$. Furthermore, given that $h^T$ solves \eqref{eq: pde-h}, long but straightforward
calculations using \eqref{eq: eign eqn} show that $v^T$ solves \eqref{eq: HJB}. Moreover, $v^T \in C^{1,2}((0,T)\times E)$ because $\hat{v}\in C^2(E)$
and $h^T\in C^{1,2}((0,T)\times E)$. This proves $(i)$.

As for part $ii)$, by applying
Proposition \ref{prop: Ito-Girsanov result} to $w=v^T,\lambda = 0$ it follows by evaluating
\eqref{eq: primal dual as eq} at $t=t,s=T$ that for the
portfolio in \eqref{eq: finite opt} and the process $M^{\eta, {v^T}}$
  from \eqref{eq: eta_w} (recall the definition of $D^w$ given in \eqref{eq: D_hatpw_p})
\begin{equation}\label{eq: opt wealth integ value}
\expec^{\prob^{\xi,t}}\bra{\frac{1}{p}\left(X^{\pi^T}_T\right)^p} =
\frac{x^p}{p}\left(v^T(t,y)\right)^{\delta}\expec^{\prob^{\xi,t}}\bra{D^{v^T}_{t,T}},
\end{equation}
since $v^T(T,y) = 1$.  In a similar manner
\begin{equation*}
\frac{x^p}{p}\expec^{\prob^{\xi,t}}\bra{\left(M^{\eta,{v^T}}_T\right)^q}^{1-p} =
\frac{x^p}{p}\left(v^T(t,y)\right)^{\delta}\expec^{\prob^{\xi,t}}\bra{D^{v^T}_{t,T}}^{1-p}.
\end{equation*}
Here $\xi= (x,y)$ and $(\prob^{\xi, t})_{\xi}$ is the solution to the martingale problem for $L$ in the canonical state space whose coordinate process starts from time $t$.
Therefore, thanks to duality results for power utility between payoffs and stochastic discount factors \cite[Lemma 5]{Guasoni-Robertson}, the claims will follow if $D^{v^T}$ is a $\prob^y$
martingale for all $y\in E$. It suffices to show $1 =
\expec^{\prob^y}\bra{D^{v^T}_T}$. It follows from \eqref{def : v^T} that
\begin{equation}\label{eq: log der rel}
\frac{v^T_y}{v^T} = \frac{\hat{v}_y}{\hat{v}} + \frac{h^T_y}{h^T}.
\end{equation}
Using this, the $\prob^y$ independence of $Y$ and $B$ implies \cite[Lemma 4.8]{Karatzas-Kardaras}
\begin{equation}\label{eq: D v^T temp 1}
\begin{split}
&\expec^{\prob^y}\bra{D^{v^T}_T} \\
&=\expec^{\prob^y}\bra{\Exp\left(\int\left(-q\Upsilon'\Sigma^{-1}\mu + A\frac{v^T_y}{v^T}\right)'\frac{1}{a}dW -
  q\int\left(\Sigma^{-1}\mu +
    \Sigma^{-1}\Upsilon\delta\frac{v^T_y}{v^T}\right)'\sigma\bar{\rho}dB\right)_T},\\
&=\expec^{\prob^y}\bra{\Exp\left(\int\left(-q\Upsilon'\Sigma^{-1}\mu + A\left(\frac{\hat{v}_y}{\hat{v}} +
  \frac{h^T_y}{h^T}\right)\right)'\frac{1}{a}dW -
  q\int\left(\Sigma^{-1}\mu +
    \Sigma^{-1}\Upsilon\delta\left(\frac{\hat{v}_y}{\hat{v}} +
  \frac{h^T_y}{h^T}\right)\right)'\sigma\bar{\rho}dB\right)_T},\\
&= \expec^{\prob^y}\bra{\Exp\left(\int\left(-q\Upsilon'\Sigma^{-1}\mu + A\left(\frac{\hat{v}_y}{\hat{v}} +
  \frac{h^T_y}{h^T}\right)\right)'\frac{1}{a}dW\right)_T},\\
&=\expec^{\prob^y}\bra{\Exp\left(\int\left(-q\Upsilon'\Sigma^{-1}\mu + A\left(\frac{\hat{v}_y}{\hat{v}} +
  \frac{h^T_y}{h^T}\right)\right)'\frac{1}{a}dW -
  q\int\left(\Sigma^{-1}\mu +
    \Sigma^{-1}\Upsilon\delta\frac{\hat{v}_y}{\hat{v}}\right)'\sigma\bar{\rho}dB\right)_T}.\\
\end{split}
\end{equation}
Note that for $w=\hat{v}$ the process of \eqref{eq: D_hatpw_p} specifies to
\begin{equation}\label{eq: D hatv temp 1}
D^{\hat{v}}_t = \Exp\left(\int\left(-q\Upsilon'\Sigma^{-1}\mu + A\frac{\hat{v}_y}{\hat{v}}\right)'\frac{1}{a}dW -
  q\int\left(\Sigma^{-1}\mu +
    \Sigma^{-1}\Upsilon\delta\frac{\hat{v}_y}{\hat{v}}\right)'\sigma\bar{\rho}dB\right)_t.
\end{equation}
This is precisely the stochastic exponential that changes the dynamics from $\prob^y$ to those for $\hat{\prob}^y$. It follows from part (ii) of Lemma~\ref{lemma: WP hat-P} and the backward martingale theorem \cite[Remark
2.3.2]{Cheridito-Filipovic-Yor} that $D^{\hat{v}}$ is a $(\prob^y, (\F_t)_{t\geq 0})$ martingale, whence
\begin{equation}\label{eq: p phat rn}
\frac{d\hat{\prob}^y}{d\prob^y}\bigg|_{\F_t} = D^{\hat{v}}_t.
\end{equation}
Furthermore, the Brownian motion $\hat{W}$ from \eqref{eq: phat sde} is
related to $W$ by $d\hat{W}_t = dW_t + (q\rho\sigma'\Sigma^{-1}\mu -
a\hat{v}_y/\hat{v}) dt$. Using this, for all $t\leq T$
\begin{equation}\label{eq: D v^T temp}
\begin{split}
&\Exp\left(\int\left(-q\Upsilon'\Sigma^{-1}\mu + A\left(\frac{\hat{v}_y}{\hat{v}} +
  \frac{h^T_y}{h^T}\right)\right)'\frac{1}{a}dW -
  q\int\left(\Sigma^{-1}\mu +
    \Sigma^{-1}\Upsilon\delta\frac{\hat{v}_y}{\hat{v}}\right)'\sigma\bar{\rho}dB\right)_t\\
&\qquad = D^{\hat{v}}_t \Exp\pare{\int a\frac{h^T_y}{h^T}d\hat{W}}_t
 = D^{\hat{v}}_t\frac{h^T(t,Y_t)}{h^T(0,y)}.
\end{split}
\end{equation}
The last equality follows from the fact that $h^T$ solves the differential
expression in \eqref{eq: pde-h} combined with Ito's formula. The second to last equality follows from the identity for any adapted, integrable processes $a,b$ and Wiener process $W$ that
\begin{equation*}
\Exp\left(\int (a_s+b_s) dW_s\right) = \Exp\left(\int a_s
  dW_s\right)\Exp\left(\int b_sd W_s-\int b_sa_sds\right).
\end{equation*}
Using \eqref{eq: D v^T temp} and \eqref{eq: p phat rn} in \eqref{eq: D v^T temp 1} and applying Proposition \ref{prop: v^T classical} it holds that
\[
\expec^{\prob^y}\bra{D^{v^T}_T} =
\expec^{\hat{\prob}^y}\bra{\frac{h^T(T,Y_T)}{h^T(0,y)}} = 1,
\]
which is the desired result.
\end{proof}

\subsection{Conditional densities and wealth processes}\label{subsec: densities conv}

The last prerequisite for the main result is to relate the terminal wealths resulting from using the finite horizon optimal strategies $\pi^T$ of \eqref{eq: finite opt} and the long-run optimal strategy
$\hat{\pi}$ of \eqref{eq: long-run opt}. Recall the definition of $D^{w}$ from
\eqref{eq: D_hatpw_p}, and consider $w= v^T$ and $w= \hat{v}$. A similar calculation to \eqref{eq: D v^T temp} using \eqref{eq: log der rel}
and \eqref{eq: D hatv temp 1} gives
\begin{equation*}
D^{v^T}_t = D^{\hat{v}}_t\frac{h^T(t,Y_t)}{h^T(0,y)}\Exp\pare{- \int a
  \frac{h^T_y}{h^T}\Delta' d\hat{B}}_t,
\end{equation*}
where
\begin{equation}\label{eq: Delta def}
\Delta = q\delta\rho'\bar{\rho},
\end{equation}
and the Brownian Motion $\hat{B}$ is
from \eqref{eq: phat sde} and related to $B$ by $\hat{B} = B +
q\bar{\rho}\sigma'\Sigma^{-1}\mu + \Delta'a\hat{v}_y / \hat{v}$. Dividing by
$D^{\hat{v}}_t$ gives
\begin{equation}\label{eq: Z ratio t}
\frac{D^{v^T}_t}{D^{\hat{v}}_t} = \frac{h^T(t, Y_t)}{h^T(0,y)} \, \Exp\pare{- \int a
  \frac{h^T_y}{h^T}\Delta' d\hat{B}}_t.
\end{equation}
For $w= v^T$ and $\lambda = 0$, \eqref{eq: primal dual as eq} gives (since all assumptions hold) almost surely $\prob^y$ (and hence $\hat{\prob}^y$) for
any $0\leq t\leq T$
\begin{equation}\label{eq: wealth-finite}
\pare{X^{0,T}_t}^p = \pare{X^{\pi, v^T}_t}^p =
x^pD^{v^T}_{t} \pare{\frac{v^T(0,y)}{v^T(t,Y_t)}}^{\delta} =
x^pD^{v^T}_{t}e^{\delta\lambda_c t} \pare{\frac{\hat{v}(y)h^T(0,y)}{\hat{v}(Y_t)h^T(t,Y_t)}}^{\delta},
\end{equation}
where the last equality uses \eqref{def : v^T}. Similarly, for
$w=\hat{v}$ and $\lambda = \lambda_c$, it
follows that for the long-run optimal strategy $\hat{\pi}$ defined in
\eqref{eq: long-run opt},\eqref{eq: primal dual as eq} gives the $\prob^y$ ($\hat{\prob}^y$) almost sure
equality, for each $0\leq t\leq T$
\begin{equation}\label{eq: wealth-long-run}
   \pare{\hat{X}_t}^p = \pare{X^{\hat{\pi}}_t}^p
   = x^pD^{\hat{v}}_{t}e^{\delta\lambda_c t} \pare{\frac{\hat{v}(y)}{\hat{v}(Y_t)}}^{\delta}.
\end{equation}
Therefore, \eqref{eq: wealth-finite}, \eqref{eq: wealth-long-run} and
\eqref{eq: Z ratio t} imply
\begin{equation}\label{eq: wealth ratio t}
\frac{X^{0,T}_t}{\hat{X}_t} =
\left(\frac{D^{v^T}_{t}}{D^{\hat{v}}_{t}}\right)^{1/p}\left(\frac{h^T(t,Y_t)}{h^T(0,y)}\right)^{-\delta/p}
= \pare{\frac{h^T(t, Y_t)}{h^T(0,y)}}^{\frac{1-\delta}{p}}\Exp\pare{-\int a \Delta \frac{h^T_y}{h^T} \, d\hat{B}}_t^{\frac1p}.
\end{equation}
where the last equality uses \eqref{eq: Z ratio t}.  Equations \eqref{eq: Z ratio t} and \eqref{eq: wealth ratio t} will be used in the next section.

\begin{rem}
The proof of Proposition \ref{lemma: HJB verification} showed $D^{v^T}$ is a
$(\prob^y,(\F_t)_{0\leq t\leq T})$ martingale for each $y\in E$. Thus, \eqref{eq: opt wealth integ value} and
  \eqref{eq: wealth-finite} in conjunction with \eqref{def: p^T} implies that
\begin{equation}\label{eq: p p^t rn}
D^{v^T}_t = \frac{d\prob^{T,y}}{d\prob^y}\bigg|_{\F_t}.
\end{equation}
\end{rem}

\subsection{Proof of main results in Section~\ref{subsec: turnpike for diffusion}}\label{subsec: proof diffusion}


\begin{proof}[Proof of Lemma \ref{lem: mart prob}]

By \cite[Theorem 5.1.5]{Pinsky}, Assumptions \ref{ass: regular coeffs} and
\ref{ass: WP P} ensure a solution, $(\prob^y_Y)_{y\in E}$ to the martingale problem for $\widetilde{\cL}$ on $E$, where
\begin{equation*}
\widetilde{\cL} = \frac{1}{2}A\partial^2_{yy} + b\partial_y.
\end{equation*}
Let $\xi = (z,y)\in\Real^{d}\times E$. The result now follows by considering
the family of measures $(\prob^{\xi} := W^{z,d}\times \prob^y_Y)_{\xi\in\Real^{d}\times E}$ on $(\Omega,\F)$
where $W^{z,d}$ is $d$-dimensional Wiener measure corresponding to a Brownian
Motion starting at $z$.

\end{proof}


\begin{proof}[Proof of Proposition \ref{prop: turnpike holds}]
By Theorem 18 in \cite*{Guasoni-Robertson}, under Assumptions \ref{ass: regular coeffs}, \ref{ass: WP P}, and \ref{ass: correl}, \eqref{eq: c decay requirement} yields the existence of a function $\hat{v}$ which satisfies
\eqref{eq: eign eqn}, \eqref{ass: recurrent}, along with the first inequality
in \eqref{ass: posrec l1}. By Holder's inequality, \eqref{eq: m prob requirement} ensures that the second inequality in \eqref{ass: posrec l1}
holds as well, proving the assertion.
\end{proof}


\begin{proof}[Proof of Lemma~\ref{thm: conv cond den}]
Recall the notation of Section \ref{subsec: densities conv}. From \eqref{eq: p phat rn},
  \eqref{eq: p p^t rn} and \eqref{eq: Z ratio t}, the limit in
 \eqref{eq: conv cond den} holds provided that\footnote{The notation $\prob^{y}\text{- lim}$ is short for the limit in probability}:
 \begin{equation}\label{eq: conv ratio Z}
  \hat{\prob}^y \text{-} \lim_{T\rightarrow \infty} \frac{h^T(t,
    Y_t)}{h^T(0,y)} \, \Exp\pare{- \int a \Delta \frac{h^T_y}{h^T} \,
    d\hat{B}}_t=1 .
 \end{equation}
where $\Delta$ is from \eqref{eq: Delta def}. Set $L^T_t = h^T(t,Y_t)/h^T(0,y)$.  Proposition \ref{prop: v^T classical}
implies that a) for each $T$, $L^T$ is a positive $\hat{\prob}^y$ martingale on
$[0,T]$ with expectation $1$ and b) for each $t\geq 0$, $\lim_{T\rightarrow\infty} L^T_t = 1$
almost surely $\hat{\prob}^y$.  Therefore, Fatou's lemma gives $1\geq \lim_{T\to \infty}
 \expec^{\hat{\prob}^y}[L_t^T] \geq \expec^{\hat{\prob}^y} [\liminf_{T\to \infty}
 L^T_t] =1$, which implies $\lim_{T\to \infty}
 \expec^{\hat{\prob}^y}\bra{|L_t^T-1|} =0$ by Scheff\'{e}'s lemma. As shown in \eqref{eq: D v^T temp}, $L^T_t = \Exp\pare{\int a h^T_y / h^T
  d\hat{W}}_t$. Lemma \ref{lemma: stoch-log conv} thus yields
\begin{equation*}
\hat{\prob}^y\ -\ \lim_{T\rightarrow
   \infty} \bra{\int a \frac{h^T_y}{h^T} \, d\hat{W}, \int a \frac{h^T_y}{h^T}
   \, d\hat{W}}_t =0.
\end{equation*}
Observing that $\|\Delta\|^2$ is a constant
 (by Assumption~\ref{ass: correl}), the previous identity implies that
 $\hat{\prob}^y$-$\lim_{T\rightarrow \infty} \bra{\int a \Delta h^T_y/ h^T
   \, d\hat{B}, \int a \Delta h^T_y/ h^T \, d\hat{B}}_t =0$, whence $\hat{\prob}^y$-$\lim_{T\rightarrow \infty} \int_0^t a \Delta h^T_y/h^T
 \, d\hat{B}=0$, which implies $\hat{\prob}^y$-$\lim_{T\rightarrow \infty}\Exp\pare{\int a
   \Delta h^T_y/ h^T \, d\hat{B}}_t =1$, i.e., the second term on the
 left-hand-side of \eqref{eq: conv ratio Z} also converges to $1$. This
 concludes the proof of \eqref{eq: conv ratio Z}.
\end{proof}


\begin{proof}[Proof of Theorem~\ref{thm: power 1-factor}]
 Let $\mathcal{S}_T$ be either $\{\sup_{u\in [0,t]} \left|r^T_u - 1\right| \geq \epsilon\}$ or $\{\bra{\Pi^T, \Pi^T}_t \geq \epsilon\}$, which are both $\F_t$-measurable. It follows from Theorem~\ref{thm: opt-port conv} and Remark~\ref{rem: finite time} part iii) that
 \begin{equation}\label{eq: conv prob p^T}
\lim_{T\to \infty}
\expec^{\hat{\prob}^y}\bra{\left.\frac{d\prob^{T,y}}{d\hat{\prob}^y}
  \right|_{\F_t} \, \indic_{\mathcal{S}_T}} =0.
\end{equation}
 On the other hand, \eqref{eq: conv cond den} and Scheff\'{e}'s lemma combined imply that
 \begin{equation*}
 \lim_{T\rightarrow \infty}
 \expec^{\hat{\prob}^y} \bra{\left|\left.\frac{d\prob^{T,y}}{d\hat{\prob}^y}\right|_{\F_t} - 1\right|} =0.
\end{equation*}
Hence, combining the previous identity with \eqref{eq: conv prob p^T}, it follows that $\lim_{T\to \infty} \hat{\prob}^y(\mathcal{S}_T) =  0$. Since  $\hat{\prob}^y \sim \prob$ on $\F_t$ from Proposition
 \ref{prop: Ito-Girsanov result}, it follows that
 \[
  \lim_{T\to \infty} \prob^y(\mathcal{S}_T) = \lim_{T\to \infty}\expec^{\hat{\prob}^y} \bra{\left.\frac{d\prob^y}{d\hat{\prob}^y}\right|_{\F_t}\, \indic_{\mathcal{S}_T}} =0,
 \]
 where the last equality follows from the dominated convergence theorem.
\end{proof}


\begin{proof}[Proof of Theorem~\ref{cor: power 1-factor myopic}]
A similar argument to the one in the proof of Lemma \ref{thm: conv cond den}, combined with \eqref{eq: wealth ratio t}, yields that $\hat{\prob}^y$-$\lim_{T\to \infty} X^{0,T}_t / \hat{X}_t =1$. On the
 other hand, Theorem~\ref{thm: power 1-factor} part a), combined with the
 equivalence between $\prob$ and $\hat{\prob}^y$, implies that
 $\hat{\prob}^y$-$\lim_{T\to \infty} X^{1,T}_t / X^{0,T}_t =1$. Hence the last two identities combined give $\hat{\prob}^y$-$\lim_{T\to \infty} \hat{r}^T_t
 =1$. Now recall that $\hat{\pi}$ is the optimal portfolio for the logarithmic
 investor under $\hat{\prob}^y$, it then follows from the num\'{e}raire property
 of $\hat{\pi}$ that $\hat{r}^T_\cdot$ is a
 $\hat{\prob}^y$-supermartingale, which implies that $\lim_{T\to \infty}
 \expec^{\hat{\prob}^y} \bra{|\hat{r}^T_t-1|} =0$, by Fatou's lemma and Scheff\'{e}'s lemma. As a result, the statements follow applying Lemma~\ref{lemma: stoch-log conv} under the probability $\hat{\prob}^y$, and remain valid under the equivalent probability $\prob^y$.
\end{proof}

\bibliographystyle{plainnat}
\bibliography{biblio}
\end{document}